\documentclass[letterpaper,12pt]{article}

\usepackage{arxiv}

\usepackage[utf8]{inputenc} 
\usepackage[T1]{fontenc}    
\usepackage{url}            
\usepackage{booktabs}       
\usepackage{amsfonts}       
\usepackage{nicefrac}       
\usepackage{microtype}      
\usepackage{lipsum}

\usepackage{enumerate}

\usepackage{comment}

\usepackage{color}
\usepackage{amsmath,amsfonts}
\usepackage{tikz}
\usepackage{enumerate}
\usepackage{mathtools}
\usepackage{mathabx}
\usepackage{multirow}
\usepackage{algorithm}
\usepackage[noend]{algpseudocode}
\usepackage{float}

\usepackage[caption=false]{subfig}

\usepackage{tikz}

\usepackage{longtable}
\usepackage{array}
\usepackage{booktabs}
\usepackage{url}
\usepackage{footnote}

\usepackage{cancel} 

\usepackage{amsthm}

\newcommand{\anatree}{\ensuremath{t}}

\newcommand{\sometermset}{\ensuremath{\mathbb{T}}}

\newcommand{\domain}{\ensuremath{D}}

\newcommand{\OR}{\ensuremath{\mathtt{OR}}}
\newcommand{\AND}{\ensuremath{\mathtt{AND}}}

\newcommand{\power}[1]{\ensuremath{\mathcal{P}(#1)}}

\newcommand{\attrdomain}{\ensuremath{D}}

\newcommand{\hy}[1]{\text{-}} 

\newcommand{\valuationFunc}{\alpha}

\newcommand{\labels}{\mathcal{L}}
\newcommand{\labelFunc}{\ensuremath{\mathit{root}}}
\newcommand{\labelsFunc}{\ensuremath{\mathit{labels}}}

\newcommand{\tfun}{\rightarrow}

\newcommand{\ineqpred}{\ensuremath{\mathcal{P}_\mathit{ineq}}}

\newcommand{\predType}[1]{\ensuremath{s_{#1}}}

\newcommand{\weak}{\preceq}
\newcommand{\maxweak}{\preceq_{\!M}^{}}

\newcommand{\valuationUni}{\mathit{VAL}_{\labelsFunc(t)\to\domain}}

\newcommand{\maps}[2]{\ensuremath{\mathcal{F}(#1, #2)}}

\newcommand{\distance}{\ensuremath{d}}

\newcommand{\metric}{\ensuremath{\distance}}

\newcommand\ATMFraudLabel{\prt{atm-fraud}}
\newcommand\CardSkimmingLabel{\prt{card-skimming}}
\newcommand\CardTrappingLabel{\prt{card-trapping}}
\newcommand\CashTrappingLabel{\prt{cash-trapping}}
\newcommand\TransactionReversalLabel{\prt{trans-reversal}}
\newcommand\AccessATMLabel{\prt{access-atm}}
\newcommand\ExecuteAttackLabel{\prt{execute-attack}}
\newcommand\BreakingIntoLabel{\prt{break-in}}
\newcommand\SocialEngineeringStaffLabel{\prt{social-engineer-staff}}
\newcommand\GetCredentialsLabel{\prt{get-credentials}}
\newcommand\ShoulderSurfLabel{\prt{shoulder-surf}}
\newcommand\GetPINLabel{\prt{get-pin}}
\newcommand\GetCardLabel{\prt{get-card}}
\newcommand\InstallCameraLabel{\prt{install-camera}}
\newcommand\InstallEPPLabel{\prt{install-epp}}
\newcommand\TakeCardLabel{\prt{take-card-phys}}
\newcommand\SocialEngineerOwnerLabel{\prt{social-engineer-owner}}
\newcommand\InstallSkimmerLabel{\prt{install-skimmer}}
\newcommand\CloneCardLabel{\prt{clone-card}}
\newcommand\StealCardLabel{\prt{steal-card}}

\newcommand{\validin}{\vdash}

\newcommand{\prt}[1]{\ensuremath{\operatorname{\mathit{#1}}}}

\newcommand\moneyAccountLabel{\prt{money-account}}
\newcommand\creditCardLabel{\prt{card}}
\newcommand\pinLabel{\prt{pin}}
\newcommand\hackAccountLabel{\prt{hack-account}}
\newcommand\moneyATMLabel{\prt{money-atm}}

\usepackage[pdfpagelabels]{hyperref}
\hypersetup{colorlinks=true, linkcolor=blue, citecolor=blue, urlcolor=blue}

\newtheorem{remark}{Remark}
\newtheorem{example}{Example}
\newtheorem{definition}{Definition}

\newtheorem{lemma}[remark]{Lemma}
\newtheorem{theorem}[remark]{Theorem}
\newtheorem{proposition}[remark]{Proposition}
\newtheorem{corollary}[remark]{Corollary}

\newcommand{\thesem}[1]{[\![#1]\!]_t}

\usepackage{listings}
\title{Attribute Evaluation on Attack Trees  with Incomplete Information}

\author{
  Ahto Buldas\\
  Cybernetica AS\\
  Estonia\\
  \And
  Olga Gadyatskaya\\
  CSC/SnT, University of Luxembourg\\
  Luxembourg\\
  \And
    Aleksandr Lenin \\
Tallin University of Technology\\
Estonia\\
  \And
  Sjouke Mauw\\
  CSC/SnT, University of Luxembourg\\
  Luxembourg\\
  \And
  Rolando Trujillo-Rasua\thanks{Corresponding author}\\
  Deakin University\\
  Australia\\
}

\begin{document}
\maketitle

\begin{abstract}
Attack trees are considered a useful tool for security modelling
	because they support qualitative as well as quantitative analysis.
	The quantitative approach is based on 
	values
	associated to each node in the tree, expressing, for instance, the
	minimal cost or probability of an attack.
	Current quantitative methods for attack trees allow the analyst to,
	based on an initial assignment of values to the leaf nodes, derive
	the values of the higher nodes in the tree.
	In practice, however, it shows to be very difficult to obtain reliable
	values for all leaf nodes. The main reasons are that data is only
	available for some of the nodes, that data is available for
	intermediate nodes rather than for the leaf nodes, or even that the
	available data is inconsistent.
	We address these problems by developing a generalisation of 
	the
	standard bottom-up calculation method 
	in three ways. First, we 
	allow initial attributions of non-leaf nodes. Second, we admit additional 
	relations between attack steps beyond those provided by the  underlying 
	attack 
	tree semantics. 
	Third, we support the calculation of an approximative solution in case of
	inconsistencies. We illustrate our method, which is based on
	constraint programming, by a comprehensive case study.

\end{abstract}


\section{Introduction}\label{sec:introduction}

Attack
trees are a useful and intuitive graphical modeling language introduced 
by Bruce Schneier \cite{Schn,schneier2011secrets} in 1999. Since then, it has 
enjoyed popularity in the security industry, as well as in the
research community. Attack 
trees have been equipped with various semantics 
\cite{MaOo,jhawar2015attack,horne2017semantics} and supported by tools 
\cite{amenaza,gadyatskaya2016attack}. They have also been enhanced
with various methods for quantitative analysis
\cite{KoMaSc-2012,Aslanyan-POST-2015,kumar2015quantitative,aslanyan2016quantitative,roy2012scalable,bistarelli2006defense,DBLP:conf/gamesec/BuldasL13,DBLP:conf/gamesec/LeninB14,DBLP:conf/gamesec/LeninWC15,TUT:LeninPhD2015},
 which allow determining for a given attack tree, for example, an 
 organisation's losses due to an 
attack, the probability that such an attack succeeds, or the cost of a
successful attack \cite{hong2017survey}.

 The underlying assumption upon which all these quantification methods are 
based is similar to the popular \emph{divide and 
	conquer} paradigm, in which a problem is recursively broken down into 
	smaller problems 
	that 
are theoretically simpler to reason about and solve.

A quantification method for attack trees often reduces to the assignment of 
attribute values to basic attack steps (leaf nodes in 
the tree). Such assignments are used in a bottom-up propagation manner to 
determine the value at the root 
node, which is a quantification measure for the scenario expressed in the 
tree~\cite{KoMaSc-2012}. 
It is largely believed that it is relatively easy to assign a reliable 
attribute value to a basic attack step, which is precise and refined enough.
Popular tools operating with attack trees, such as ADTool 
\cite{KordyKMS-QEST-13,gadyatskaya2016attack} and SecurITree \cite{amenaza}, 
work exactly under 
this premise.

In practice, the assumption that attribute values for more concrete attack 
steps 
are easier to obtain has proven incorrect. 
Indeed, most companies manage to obtain statistical data for abstract 
attacks, e.g.\ frequency of skimming attacks, while they might struggle to come 
up 
with similar data for more refined attacks, e.g.\ frequency of stereo skimming 
attacks based on audio technology. For security consultants, it might be 
feasible to obtain reliable estimations for (at least some) abstract attacks in 
relevant domains, but precise historical data for low-level attack steps might 
be out of reach. Thus, we observe that there is a tension between the limited 
availability of data and the requirement to provide data values for all leaves 
in an attack tree before proceeding with a quantification method.

Today, existing quantitative approaches for attack trees cannot handle values 
of intermediate nodes in the tree that may become available from historical 
data. Moreover, they do not support the use of additional constraints over 
nodes in the tree, which are obtained from external sources of information 
rather than from the attack tree model itself. For example, the analysts may be 
confident that card skimming attacks are more 
frequent than physical attacks on card holders. Such a relation cannot be 
captured in an attack tree model, because it is not a hierarchical relation, 
hence it is ignored in current quantitative approaches for attack trees.

There is clearly a need for novel computation methods on attack trees
that account for available historical data and domain-specific
knowledge. In this paper, we formulate a general \emph{attack-tree
decoration problem} that treats assignment of values to tree nodes as
a problem of finding a set of data values satisfying a set of
predicates. These predicates arise from the attack tree structure and
the target attribute to be computed (i.e.\ semantics) and from the
attainable historical values and domain knowledge (i.e.\ available
data). Our methodology to solve the attack-tree decoration problem
accounts for scenarios in which the set of predicates cannot be jointly
satisfied, due to inconsistencies or possible noise in the data.

\noindent\emph{Contributions.} 
In this work 

\begin{itemize}
	\item We transform an attack tree semantics together with an attribute 
	interpretation into a constraint satisfaction problem 
	(Section~\ref{sec-decoration}). If
the attack tree semantics is consistent with the attribute
interpretation, this allows us to determine appropriate attribute
values for all nodes in the tree.  
	\item Because confidence in the available historical data 
and domain-specific knowledge may vary, we provide a methodology to deal with 
inconsistencies (Section~\ref{sec-methodology}). The usefulness of our 
approach is that any consistent valuation is better than no valuation, as it 
will enable the follow up process of using the attack tree for what-if 
analysis. The standard bottom-up approach would result in absence of any 
valuation until all leaf node values can be assigned. 
	\item We introduce two concrete approaches to deal with inconsistencies 
	(Section~\ref{sec:sec5}). The first one determines 
the smallest subset of constraints that makes the decoration problem 
inconsistent, which is useful to find contradictory or wrong assumptions. The 
second one is suitable for constraints that are expressed in the form 
of inequalities. In this approach constraints are regarded consistent and an 
optimal decoration is always found. The proposed methodology has been 
implemented as proof-of-concept 
	software tools\footnote{Code available at \url{https://github.com/vilena/at-decorator}}.
	\item We validate our methodology and the implementations through a 
	comprehensive case study on the security of 
	Automatic Teller Machines 
	(Sections~\ref{sec-evaluation}-\ref{sec:empirical}).
\end{itemize}

\section{Related Work}
\label{sec-background}

Research articles on quantitative security analysis with attack trees in all 
their flavors (attack-defense trees, defense trees, etc. \cite{kordy2014dag}) 
often focus on \emph{providing extensions to attack trees} enabling more complex 
scenarios 
\cite{kordy2012computational,aslanyan2016quantitative,jhawar2016stochastic,gadyatskaya2016modelling,Arnold-POST14}
 and \emph{defining metrics for evaluating scenarios} captured as attack 
trees. Various metrics have been considered in the literature, for
instance, the probability/likelihood of an attack
\cite{BaKoMeSc,TUT:LeninPhD2015}, expected time until a successful attack \cite{BaKoMeSc}, attacker's utility \cite{lenin2014attacker,DBLP:conf/gamesec/BuldasL13,DBLP:conf/gamesec/LeninB14,DBLP:conf/gamesec/LeninWC15,TUT:LeninPhD2015}, 
return on security investment 
\cite{bistarelli2006defense}, or assessment of risks \cite{gadyatskaya2016bridging,potteiger2016software}. Bagnato et al.~\cite{BaKoMeSc} present a list of metrics found in security literature that can be computed on attack trees.
Yet, all these 
 approaches assume 
 that the data values to perform 
quantitative analysis of system security are readily available. Indeed, to the 
best of our knowledge, no methodologies have been developed for integrating 
historical data and domain-specific knowledge in quantitative analysis of 
attack trees. At the same time, even the inventor of attack trees Bruce 
Schneier has acknowledged the painstaking work for data collection that is a 
prerequisite for quantitative analysis on the trees \cite[Chap. 
21]{schneier2011secrets}.

Benini and Sicari~\cite{benini2008risk} have proposed a framework for attack tree-based security risk 
assessment. The approach relies on identification of 
security vulnerabilities, that are placed in the leaves of an attack tree. 
Quantitative parameters of the vulnerabilities, such as exploitability and 
damage potential, allow to estimate security risks to a system. In 
\cite{benini2008risk}, the exploitability parameters are initially evaluated 
based on the CVSS scores\footnote{Common Vulnerability Scoring System 
\url{https://www.first.org/cvss}} of the respective vulnerabilities, and then 
they are adjusted based on the expert judgement about the system context and 
mutual effect of vulnerabilities on each other expressed in a vulnerability 
dependency graph.
While this methodology offers more precise quantitative risk 
assessment with attack trees, it assumes that the attack tree is  constructed in 
a bottom-up manner. All system vulnerabilities have to be identified using a suitable 
technology, and accommodated in an attack tree. In complex environments it will likely be impractical to 
apply the bottom-up approach due to the huge amount of potential vulnerability 
combinations that can be exploited in various attacks.
Indeed, in practice attack trees are typically
designed in a top-down manner, when the analyst starts by conjecturing the 
main attacker's goal and  iteratively breaks it down into smaller subgoals 
\cite{MaOo,Schn,PoEM-2016,schneier2011secrets}. 

Recently, de Bijl~\cite{de2017using} studied the use of historical
data values to obtain attribute values for attack tree nodes. He
proposed several heuristics to deal with missing data values,
including the standard bottom-up algorithm to infer parent node
values, the reuse of values for recurring nodes, and the use of
various data sources to estimate certain attributes. For example, the
paper refers to \emph{distance to the police station} as a hidden
variable influencing \emph{probability} of attack.
Yet,~\cite{de2017using} does not define a precise methodology to
perform computations on attack trees with missing leaf node values.

\noindent\emph{Quantitative analysis of fault trees.}
Fault trees are close relatives of attack trees that are widely used in the reliability domain. There exists a large body of literature dedicated to quantitative analysis of fault trees. Yet, the standard bottom-up approach in attack trees is not the standard approach in fault trees, where min-cut set analysis and translation into more intricate models are common~\cite{ruijters2015fault}.

Fuzzy fault trees address evaluation of fault trees under uncertainty
\cite{mahmood2013fuzzy}, when failure statistics are not fully
available. The main difference is that fuzzy fault trees have been
developed to serve the needs of the reliability community and fault
tree application methods (fuzzy probability functions, error
propagation estimates, etc.), while the security community and attack
tree application methods have different needs (bottom-up computation
for a large variety of attributes). Therefore, solutions designed for
fault trees do not fully address the problems in the attack tree
space. 

\noindent\emph{Data issues in quantitative risk assessment.}
Attack trees are typically used for threat modeling and security risk 
assessment \cite{shostack2014threat}. Thus, it is necessary to evaluate the 
data availability perspective also in the more general context of security risk 
assessment. 
Indeed, the general 
question of data validity in quantitative security risk assessment (QSRA for 
short) and the reliability of QSRA results in presence of uncertainty in data 
values has been raised by many practitioners and researchers \cite{vose2008risk}. 

QSRA enables decision making based on quantitative estimations of some relevant 
variables (e.g.\ probability of an event, cost, time, vulnerability, etc.). 
These quantitative estimations are typically aggregated in a model
that can then 
be utilised by a decision maker \cite{vose2008risk}. Many studies, books on 
security, and industry reports have acknowledged that the quality of 
quantitative risk analysis, and, correspondingly, the decisions made based on 
it, heavily depends on the quality of data used 
\cite{vose2008risk,jaquith2007security,schneier2011secrets,baker2007necessary}. 
Notably, it has been established that probabilities of particular loss 
events and 
costs associated to security spending can be hard to obtain from
historical data 
\cite{jaquith2007security,oppliger2015quantitative,aven2007unified,bohme2010security,ahmed2007review}. This body of knowledge serves as evidence of inherent difficulty to obtain meaningful estimates for probability and cost of detailed attack steps, i.e., values for leaf nodes in attack trees.

Nevertheless, it has been acknowledged that, for instance, for
insurance companies it might be feasible to get meaningful data,
because they have access to an entire population, i.e.\ they have good
statistics \cite{oppliger2015quantitative}. It has also been
demonstrated, for example, that breach statistics can be used to
predict future breaches in different segments
\cite{sarabi2015prioritizing}, and that statistics pertinent to
different user profiles can be applied to estimate success rates of
intrusions \cite{dacier1996models}.

Furthermore, for security assessment, it has been long established that external data sources, such as threat level indicators (e.g. malware numbers) can be helpful to update quantitative risk assessment models \cite{bohme2010security}. Therefore, enabling better usage of available historical data, which may not directly correspond to information about low-level attack events (leaf nodes), will be a valuable enhancement for quantitative analysis of attack trees. 
From this review of the relevant scientific literature, we can conclude that there is a strong need for an approach to perform quantitative analysis on attack trees in case the analyst cannot confidently assign values to all leaf nodes. Furthermore, this approach needs to integrate available historical data that can come in form of values for some abstract attacks (intermediate nodes) or constraints (equalities and inequalities) on combinations of attack tree node values. In the remainder of this paper, we propose such an approach.

\section{Attack-Tree Decoration }
\label{sec-decoration}

In this section we give, to the best of our knowledge, the first formulation of 
the attack-tree decoration problem as a constraint satisfaction problem. We 
start by introducing the necessary attack tree basics. The interested reader 
can find more details about the attack tree formalism in the paper by Mauw and 
Oostdijk \cite{MaOo}.

\subsection{Attack Trees}

In an attack tree the main goal of the attacker is captured by
the root node. This goal is then iteratively refined into subgoals,
represented by the children of the root node. Leaf nodes in 
an attack tree are called \emph{atomic subgoals}, as they are not refined 
any further. 
Non-leaf nodes, instead, can be of two types: disjunctive ($\OR$) or 
conjunctive ($\AND$).  A conjunctive refinement expresses that all subgoals 
must to be achieved in order to succeed on the main goal, while in a 
disjunctive refinement the achievement of a single subgoal is already enough. 

\begin{figure}[!ht]
	\centering
	\includegraphics[width=0.35\textwidth]{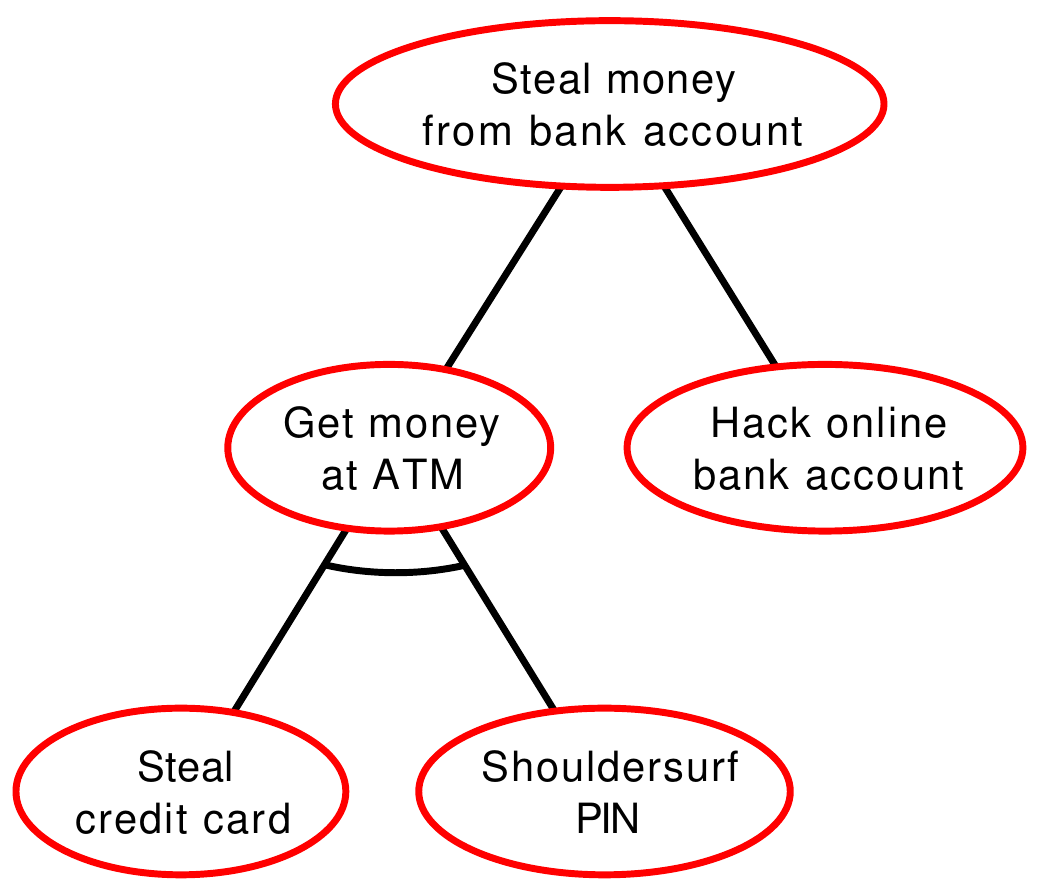}
	\caption{An attack tree representing stealing money from someone's bank 
	account. 
		\label{fig:small_tree}}
\end{figure}

\begin{example}\label{example:small-tree}
Consider the simple attack tree in Figure~\ref{fig:small_tree}. The
root node of this tree represents the main goal of the attack: to
steal money from a bank account. This goal is disjunctively refined
into two alternative sub-attacks: the attacker may try to get money
from an automated teller machine (ATM), \emph{or} they might attempt
to hack the online bank account system. The sub-goal that explores
getting money at an ATM is further conjunctively refined into two
complementary activities: the attacker must steal the credit card of
the victim \emph{and} they also needs to obtain the PIN code by
shoulder-surfing. Note that a conjunctive refinement is denoted
graphically with an arc spanning the child nodes. 
\end{example}

\begin{definition}[Attack tree]\label{def-atree}
	Given a set of labels $\labels$, an attack tree (ATree)
	is constructed according to the following 
	grammar (for $\ell\in\labels$):
	\[\anatree ::= \ell \mid
	\OR(\anatree,\ldots,\anatree)_{\ell}
	\mid
	\AND(\anatree,\ldots,\anatree)_{\ell}
	\text{.}
	\]
	
\end{definition}

Our grammar above slightly differs from the grammar used 
in other notations to represent attack trees~\cite{MaOo,KoMaRaSc_JLC},
as we require every node 
in the tree to be annotated with a label $\ell$. The reason for this is that, 
as opposed
to standard attack tree semantics that focus on the leaf nodes, we render every 
node in the tree equally important. 

To provide a definition of our running 
example we will use shorter labels than those in Figure~\ref{def-atree}. 
The actual mapping between labels should become clear through a quick visual 
inspection.

\[\OR(\AND(\creditCardLabel, \pinLabel)_{\moneyATMLabel}, 
\hackAccountLabel)_{\moneyAccountLabel}.\]

We say that an attack tree has \emph{unique labels} if it does not
contain two distinct nodes with the same label.
We use $\sometermset$ to denote the universe of 
attack trees. We also use the auxiliary functions $\labelFunc \colon 
\sometermset \rightarrow \labels$ and $\labelsFunc \colon 
\sometermset \rightarrow \power{\labels}$ to obtain, respectively,  
the 
root node's label and all labels of a given tree. Formally,  
\begin{itemize}
	\item $\labelFunc(\anatree) 
= \ell \iff 
\anatree \equiv \ell \vee \anatree \equiv                
\OR(\anatree_1,\ldots,\anatree_n)_{\ell} \vee \anatree \equiv
\AND(\anatree_1,\ldots,\anatree_n)_{\ell}$ for some $t_1, \ldots, t_n \in 
\sometermset$
	\item $\labelsFunc(\anatree) = \{\ell\}$ if $\anatree \equiv \ell$, 
	otherwise $\labelsFunc(\anatree) = \{\ell\} \cup \labelsFunc(\anatree_1) 
	\cup \dots \cup \labelsFunc(\anatree_n)$ when $\anatree 
	\equiv                
\OR(\anatree_1,\ldots,\anatree_n)_{\ell} \vee \anatree \equiv
\AND(\anatree_1,\ldots,\anatree_n)_{\ell}$ for some $t_1, \ldots, t_n \in 
\sometermset$

\end{itemize}

For example, given the tree from Fig.~\ref{fig:small_tree}, we have that 
$\labelFunc(\AND(\creditCardLabel, 
\pinLabel)_{\moneyATMLabel}) = \moneyATMLabel$ and 
$\labelsFunc(\AND(\creditCardLabel, 
\pinLabel)_{\moneyATMLabel}) = \{\creditCardLabel, \pinLabel, 
\moneyATMLabel\}$.

\subsection{The Attack-Tree Decoration Problem}

We proceed by formulating the attack-tree decoration problem as a 
constraint satisfaction problem. Intuitively, we map an attack tree to a set of 
boolean expressions whose variables are drawn from the set of labels of the 
tree. Such a set of boolean expressions, defined over a given domain, can be 
seen as a constraint satisfaction problem whose solutions correspond to 
solutions of 
the attack-tree decoration problem. The remainder of this section is dedicated 
to formalising this intuition. 

Decorating an attack tree is a process whereby nodes in the tree are
assigned with values. Given an attack tree $\anatree$, we use a total
function $\valuationFunc\colon \labelsFunc(\anatree) \tfun
\attrdomain$ from labels of the tree to values in a domain
$\attrdomain$ to represent the decoration process, and $\valuationUni$ to 
denote the 
universe of such functions.
To that effect, we 
often refer to labels as variables and to $\valuationFunc$ as a 
\emph{valuation}. 
The co-domain 
$\attrdomain$ of a
valuation is determined by the attribute of the tree under
evaluation. For example, \emph{minimum time of a successful attack} uses the
natural number domain $\mathbb{N}$ to express discrete time, while
\emph{required attacker skill to succeed} typically uses a discrete and 
categorical
domain, such as $\{\textsf{low}, \textsf{medium}, \textsf{high}\}$. 

\begin{definition}[Attribute semantics]
Given an attack tree $\anatree$ and a domain $\domain$, an
\emph{attribute semantics} is a set of valuations with domain
$\labelsFunc(\anatree)$ and co-domain $\domain$.
\end{definition}

A semantics provides an attribute with the set 
of valuations that the 
attribute regards as valid in a given tree. Because defining an attribute 
semantics by exhaustive enumeration of its valuations might be cumbersome, we 
consider in this article attributes whose semantics can be derived from a 
constraint satisfaction 
problem. 

An \emph{attribute constraint} is defined as a boolean 
expression over the set of labels of 
a tree. To that effect, when we use labels in expressions we will consider them 
as variables over a given domain $\domain$. For example, if the attribute 
\emph{minimum 
time taken by 
an attack} is being computed 
over an attack tree of the form $t \equiv \OR(t_1, \ldots, t_k)_{\ell}$, it is 
typically required that $\ell = 
\min(\labelFunc(t_1), 
\ldots, \labelFunc(t_n))$~\cite{kordy2012computational}. The intuition for such 
constraint is that, 
because $\ell$ is disjunctively refined, the minimum
time needed by an attacker to meet the goal $\ell$ is 
considered to be the least time required by any of $\ell$'s children. 

We use predicates as short-hand notations for boolean expressions. For example, 
$\textsf{min-time}(\ell, \ell_1, \cdots, \ell_n)$ can be used to denote the 
boolean expression $\ell = \min(\ell_1, \ldots, \ell_n)$.  
We say that a predicate $p(\ell_1, \ldots, 
\ell_n)$ is valid under interpretation $\valuationFunc$, denoted $p(\ell_1, 
\ldots, \ell_n) \validin \valuationFunc$, if $p(\valuationFunc(\ell_1), \ldots, 
\valuationFunc(\ell_n))$ evaluates to $\texttt{true}$. 
Likewise, a set of predicates $A$ is said to be valid under 
$\valuationFunc$, denoted $A \validin 
\valuationFunc$, if all predicates in $A$ are valid under $\valuationFunc$. 
When it does not lead to confusion, we will often refer to a predicate 
$p(\ell_1, \ldots, 
\ell_n)$ as $p$.

\begin{definition}[Attribute 
constraint-set]\label{def-att-domain}
Given an attack tree $\anatree$ and a domain $\domain$, an
\emph{attribute constraint-set} is a set of predicates $\{p_1, \ldots, p_n\}$
over $\labelsFunc(\anatree)$ whose variables range over $\domain$.
Its semantics is defined by $\thesem{\{p_1, \ldots, p_n\}} =
\{\valuationFunc \in \mathit{VAL}_{\labelsFunc(t)\to\domain} \mid
\{p_1, \ldots, p_n\} \validin \valuationFunc\}$.
\end{definition}

There exist in literature various ways to relate the value of a parent node in 
an attack tree to the values at its children~\cite{kordy2012computational}, 
of which the bottom-up approach is the most common one~\cite{MaOo,KoMaRaSc_JLC}.
This bottom-up approach starts from an assignment of concrete values
to the leaf nodes of the tree and uses two functions (one for disjunctive
refinement and one for conjunctive refinement) to recursively
calculate the value of a parent node from the values of its children.
We will next define how an attribute constraint-set can be recursively
derived from two unranked aggregation operators associated with a
bottom-up approach.
The actual values of the leaf nodes will have to be defined by the
analyst through additional constraints.

\begin{definition}[Bottom-up attribute 
constraint-set]\label{def:bottom-up-attribute-domain}

Let $\anatree$ be an attack tree, and $\vee$ and $\wedge$ 
two unranked function symbols (symbols without fixed arity) with domain 
$\domain$. We use 
$\check{p}(\ell_1, \ldots, \ell_{n+1})$ to denote the boolean expression 
$\ell_1 
= 
\vee(\ell_2, \ldots, \ell_{n+1})$. Similarly, the predicate 
$\hat{p}(\ell_1, \ldots, \ell_{n+1})$ denotes the boolean expression $\ell_1 
= \wedge(\ell_2, \ldots, \ell_{n+1})$. The \emph{bottom-up attribute 
constraint-set} $P(t)$ of
$\anatree$ is recursively computed as follows:
	\begin{itemize}
    \item If $\anatree \equiv \ell$ for some label $\ell$, then
		$P(t)= \emptyset$.
		\item If $\anatree \equiv \OR(\anatree_1, \ldots, 
		\anatree_n)_{\ell}$ 
		with 
		$\labelFunc(t_i) = \ell_i$ for $i \in \{1, \ldots, n\}$, then 
		$P(t) = P(t_1) \cup \cdots \cup P(t_n) \cup \{\check{p}(\ell, 
		\ell_1, \ldots, \ell_n) \}$;
		\item If $\anatree \equiv \AND(\anatree_1, \ldots, 
		\anatree_n)_{\ell}$ with 
		$\labelFunc(t_i) = \ell_i$ for $i \in \{1, \ldots, n\}$, then 
		$P(t) = P(t_1) \cup \cdots \cup P(t_n) \cup \{\hat{p}(\ell, \ell_1, 
		\ldots, \ell_n \})$.
	\end{itemize}
\end{definition}

Definition~\ref{def:bottom-up-attribute-domain} is based on the standard 
bottom-up approach~\cite{MaOo,KoMaRaSc_JLC} for attack trees where child nodes 
are aggregated 
together based on two functions: $\wedge$ for children of a
conjunctive refinement
and $\vee$ for children of a disjunctive refinement.  In literature there exist 
concrete definitions of $\wedge$ and $\vee$ for various attributes. For 
example, when computing probability of success it is usually considered that 
$\wedge = \times$ and $\vee = +$, for cost $\wedge = +$ and $\vee = 
\min$, and for minimum time $\wedge = \max$ and $\vee = 
\min$. 

\begin{definition}[The attack-tree decoration 
problem]\label{def-decoration-problem}
Given an attack tree $\anatree$ and an attribute constraint-set $\{p_1, \ldots, 
p_n\}$ for $\anatree$, the \emph{attack-tree decoration problem}
consists in finding a valuation in $\thesem{\{p_1, \ldots, p_n\}}$.
\end{definition}

The attack-tree decoration problem corresponds to the well-known
\emph{Constraint Satisfaction Problem} (CSP)~\cite{T1993}, where a
solution is a valuation that satisfies a set of 
constraints. Finding a solution or even deciding whether there exists a 
solution for CSP is a well-known and complex computational problem.

We say that the attack-tree decoration problem is:
\begin{itemize}
	\item \emph{Determined:} If the cardinality of $\thesem{\{p_1, \ldots, 
	p_n\}}$ is one, i.e.\ there exists a single valid valuation only.
	\item \emph{Inconsistent:} If $\thesem{\{p_1, \ldots, 
	p_n\}} = \emptyset$, i.e.\ there does not exist a valid valuation.
	\item \emph{Undetermined:} If the cardinality of $\thesem{\{p_1, \ldots, 
	p_n\}}$ is larger than one, i.e.\ the problem is neither inconsistent nor 
	determined.
\end{itemize}

We illustrate these concepts with the following example that utilises a subtree 
from our running example. Consider the 
tree $\anatree$ depicted in Figure~\ref{tree1} whose set of labels is
\[
\labelsFunc(\anatree)=\{\moneyAccountLabel, \moneyATMLabel,
  \hackAccountLabel\}\text{.}
\]

\begin{figure*}[htp]
\centering

\subfloat[An example of 
undeterminism. 
\label{tree1}]{\includegraphics[width=0.3\textwidth]{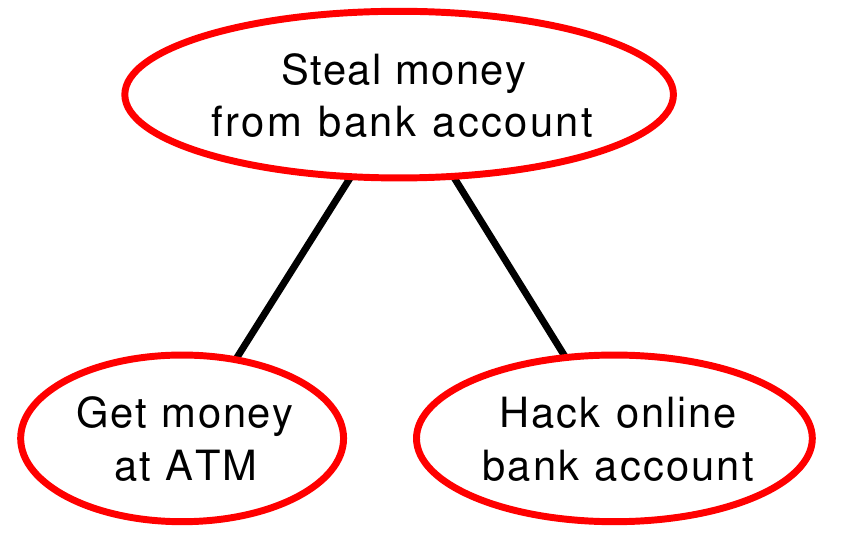}}
\hfill 
\subfloat[An example of 
inconsistency. 
\label{tree2}]{\includegraphics[width=0.3\textwidth]{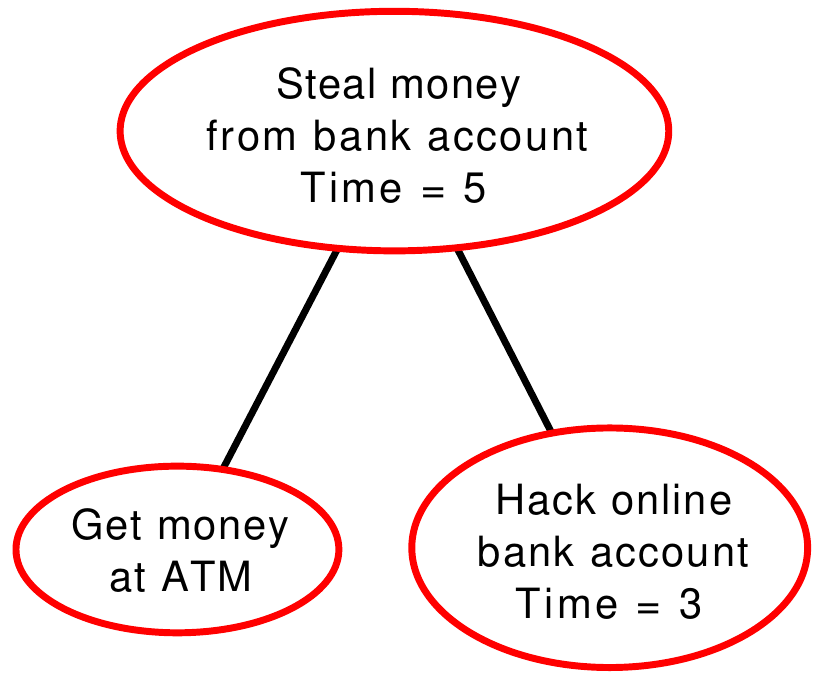}}
\hfill
\subfloat[A determined attribute 
domain. 
\label{tree3}]{\includegraphics[width=0.3\textwidth]{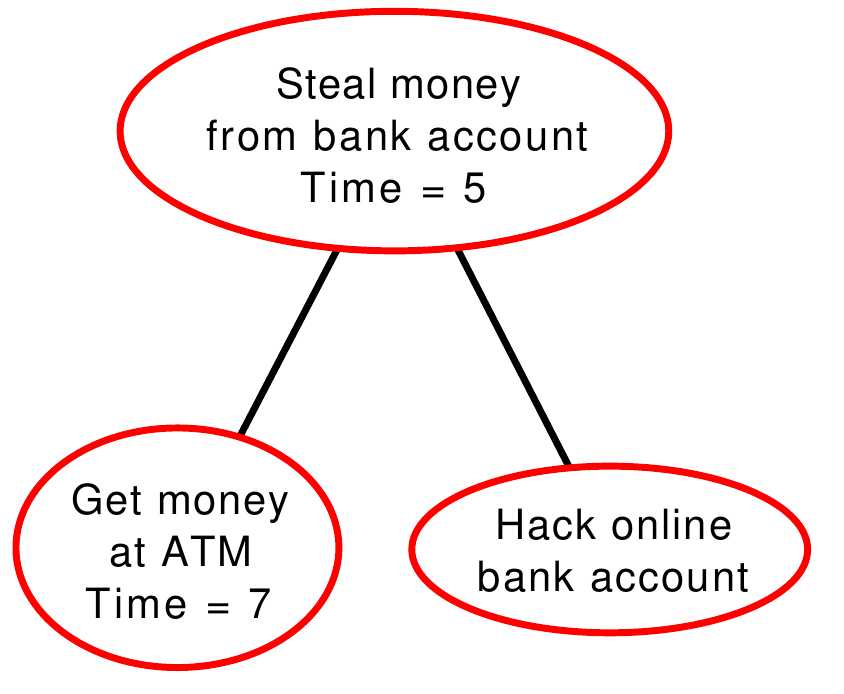}}
\caption{}
\end{figure*}

\noindent
Further, we consider subsets of the following set of constraints:
\[
\begin{array}{lll}
p_1 & :=&  \moneyAccountLabel = \min(\moneyATMLabel, \hackAccountLabel)\\
p_2 & :=&  \moneyAccountLabel = 5\\
p_3 & :=&  \hackAccountLabel = 3\\
p_4 & :=&   \moneyATMLabel = 7\\
\end{array}
\]
Predicate $p_1$ follows from
the standard interpretation of minimum attack time in an attack 
tree~\cite{KoMaSc-2012}, where $\wedge = \max$ and $\vee = \min$ defined over 
the natural numbers $\mathbb{N}$. This leads to the bottom-up 
attribute constraint-set $\{p_1\}$.
The attack-tree decoration problem with this attribute constraint-set
is clearly undetermined given that, for example,  the valuations 
$\{\moneyAccountLabel \mapsto 3, \moneyATMLabel \mapsto 3, \hackAccountLabel 
\mapsto 5\}$ and 
$\{\moneyAccountLabel \mapsto 3, \moneyATMLabel \mapsto 3, \hackAccountLabel 
\mapsto 4\}$ both satisfy $p_1$. 

Now, consider predicates $p_2$, $p_3$ and $p_4$.
This type of boolean expressions represent variable assignments. We
observe that the attribute constraint-set $\{p_1, p_2, 
p_3\}$ leads to an inconsistent decoration problem (see 
Figure~\ref{tree2}), while  
the decoration problem defined by $\{p_1, p_2, 
p_4\}$ is determined as there exists a unique 
valuation satisfying all three predicates, 
namely $\{\moneyAccountLabel \mapsto 5, \moneyATMLabel \mapsto 7, 
\hackAccountLabel \mapsto 5\}$ (see 
Figure~\ref{tree3}). Finally, we remark that the attribute constraint-set 
$\{p_1, p_3, p_4\}$ also leads to a decoration problem that is 
determined. Moreover, it corresponds to the standard bottom-up calculation in 
attack trees. 

This last example illustrates the general observation that, given an
assignment of values to the leafs of an attack tree and a bottom-up
attribute constraint-set (Def.~\ref{def:bottom-up-attribute-domain}),
the attack-tree decoration problem is determined. This is formalized
in the following proposition, which can be easily proved by induction
on the structure of the tree.

\begin{proposition}\label{prop:bottom-up-determined}
Let $\anatree$ be an attack tree with unique labels, let
$L=\{l_1,\ldots,l_n\}$ be the set of labels of its leaf nodes and let
$D$ be a domain.
Let $P_L = \{l_1=v_1,\ldots,l_n=v_n\}$ be a set of constraints
assigning values $v_1,\ldots,v_n\in D$ to the leaf nodes and
let $P(\anatree)$ be the bottom-up attribute constraint-set of $\anatree$.
Then the attack-tree decoration problem for $\anatree$ and
constraint-set $P_L \cup P(\anatree)$ is determined.
\end{proposition}

\section{A Methodology for Attack-Tree Decoration}
\label{sec-methodology}

As indicated above, our approach extends the rather rigid bottom-up
way in which attack trees are currently decorated.
Our methodology consists of two main steps that complement each other:
(1) generation of the attribute constraint-set and (2) analysis of valid
valuations. The former boils down to the definition of predicates over
the set of labels of a tree. We make a distinction between two types
of predicates: hard predicates and soft predicates.

\subsection{Hard Predicates}

Hard predicates are derived from the attack tree refinement structure rather 
than from knowledge databases or an expert's opinion. This choice
establishes that all predicates derived 
from the attack tree structure should be satisfied, as otherwise the attribute 
semantics and the tree contradict each other. 
The term \emph{hard}
stems from the notion of hard and soft constraints in satisfaction
problems. Soft constraints represent desirable properties, while hard
constraints are a must.

In this article, we consider hard predicates those contained in the bottom-up 
attribute constraint-set of the tree (see 
Def.~\ref{def:bottom-up-attribute-domain}). 
This is a conservative 
choice that allows us to extend existing bottom-up quantification methods based 
on the refinement
relation of the tree~\cite{MaOo,KoMaRaSc_JLC}. In fact, it
follows from Prop.~\ref{prop:bottom-up-determined}
that, if all labels in a tree are unique, then the resulting attack-tree 
decoration problem based on this bottom-up 
attribute constraint-set is either undetermined or determined. 
We remark, nonetheless, that our 
methodology can also be used to model other computational approaches 
such as the Bayesian reasoning proposed by Kordy, Pouly, and 
Schweitzer~\cite{KoPoSc_iFM14}. It is ultimately the analyst who decides what 
constitutes a set of hard predicates, although we require the analyst to come 
up 
with 
hard 
predicates that are satisfiable; as we do in this article.

\subsection{Soft Predicates}

Statistical data and constraints extracted from industry-relevant 
knowledge-bases and experts are too valuable to ignore. And our methodology  
treats them as first-class citizens. As usual, we encode this information in 
predicates. For example, assume that comprehensive empirical data 
indicates that the \emph{probability of a bank account being hacked} is less 
than $0.01$. The semantics of such attribute in our running example tree can be 
defined by a set containing the predicate 
$\texttt{hack-prob}(\hackAccountLabel) := 
\hackAccountLabel \leq 0.01$.

In our methodology, predicates obtained from experts and knowledge-bases are 
regarded as soft. The reason is that, when it comes to opinion and empirical 
data, 
inconsistencies are common. 
Hence we do allow these predicates to be violated 
up 
to some extent. For example, consider that for a particular attack tree we 
obtain that the 
probability that an account is hacked is $0.02$. Although such an outcome 
violates the predicate $\texttt{hack-prob}(\hackAccountLabel)$, one may find it 
acceptable and not \emph{far} from the considered empirical data.

\subsection{Analysis of Attribute Semantics with Hard and Soft Predicates}

Given an attack tree $\anatree$ and attribute constraint-set $\{p_1, \ldots, 
p_n\}$ over labels of $\anatree$ and domain $\domain$, we use $H(t)$ and $S(t)$ 
to denote the 
partition of $\{p_1, \ldots, 
p_n\}$ into 
hard and soft predicates, respectively.
As described in the previous section, we analyse an attribute constraint-set by 
looking at solutions of the corresponding constraint satisfaction problem. 
Formally, given an attribute constraint-set $H(t) \cup S(t)$, 
we aim 
at finding a 
valuation $\valuationFunc \in \valuationUni$ such that $H(t) \cup S(t) \validin 
\valuationFunc$. 
However, such formulation makes no distinction between 
hard and soft predicates, which is a feature we regard important in our 
methodology. For example, it may be the case that no valuation $\valuationFunc$ 
satisfying $H(t) \cup S(t) \validin \valuationFunc$ exists, while we can still 
find  
$\valuationFunc' \in \valuationUni$ such that $H(t) \validin \valuationFunc'$. 
Note 
that, although the constraint-set $H(t)$ is an oversimplification of the 
original attribute constraint-set with all soft constraints being removed, 
$\valuationFunc'$ 
satisfies all hard constraints and thus may be worth considering. 

In our methodology, when the original attack-tree decoration problem has no 
solution we propose to solve a weaker version: \emph{the relaxed
attack-tree decoration problem}. This new problem allows soft predicates to be 
weakened, which consists in replacing any soft predicate $p \in S(t)$ by a  
predicate $p'$ that logically follows from $p$. 
We define this type of entailment on predicates over an attack tree
$\anatree$ by:
\begin{align*}
& p\implies p' \text{~~if and only if~~} \\ 
& \quad \forall\valuationFunc\in\valuationUni\colon p \validin \valuationFunc
\implies p' \validin \valuationFunc\text{.}
\end{align*}

Using this notation, we can define the notion of a weakening relation
on sets of predicates.

\begin{definition}[Weakening relation]\label{def:weakening}
Let $\weak$ be a partial order on sets of predicates.
Then we say that $\weak$ is a \emph{weakening relation} if and only if
for all sets of predicates $P$ and $P'$ it holds that
\[P'\weak P \implies
  (\forall p'\in P'~ \exists p\in P\colon p\implies p')\text{.}
\]

\end{definition}

\noindent
If $P'\weak P$, we say that $P'$ is a weakening of $P$ under the
weakening relation $\weak$.
We provide three examples of weakening relations.
\begin{enumerate}
\item Set equality ($=$), which is the trivial weakening relation.
\item Set inclusion ($\subseteq$), which allows one to weaken a set of
predicates by deleting one or more of its elements.
\item The maximal weakening relation ($\maxweak$), which is defined by
\[P'\maxweak P \iff
  (\forall p'\in P'~ \exists p\in P\colon p\implies p')\text{.}
\]
\end{enumerate}
The proofs that these are indeed weakening relations and that
$\maxweak$ is maximal are straightforward.

Using this notion of a weakening relation we reformulate the
attack-tree decoration problem as an optimisation problem in the 
following way.

\begin{definition}[The relaxed attack-tree decoration problem]
	\label{def:relaxed_decoration_problem}
	Let 
	$\anatree$ be an attack tree and $H(t) \cup S(t)$ an 
	attribute constraint-set over $\labelsFunc(t)$ and domain 
	$\domain$, where $H(t)$ and $S(t)$ are hard and soft predicates, 
	respectively. Let $\weak$ be a weakening relation. 
	The \emph{relaxed 
		attack-tree decoration problem} consists of two stages: 
\begin{enumerate}
	\item Finding a set of predicates 
		$S$ over $\labelsFunc(t)$ and domain $\domain$ such that:
	\begin{itemize}
		\item $S \weak S(t)$,
		\item $ \thesem{H(t) \cup S}  \neq \emptyset$, and
		\item $\forall S' \colon S \weak  S' \weak S(t) 
 \implies (\thesem{H(t) \cup S'} = \emptyset \vee S' \weak S)
 $.
	\end{itemize}
	\item Solving the attack tree decoration problem 
	with constraint-set  $H(t) 
	\cup S$.
\end{enumerate}
A solution is a pair $(S, \valuationFunc)$, such that $\valuationFunc \in  
\thesem{H(t) \cup S}$. 

\end{definition}

The choice of the weakening relation is relevant in an
instantiation of 
the relaxed attack tree decoration problem, as we show in the next 
section. 
In particular, we analyse two relevant decoration problems resulting
from two concrete weakening relation definitions, namely the set
inclusion ($\subseteq$) and maximal weakening relation ($\maxweak$).  

\section{Decoration Algorithm for Specific Classes of Predicates}
\label{sec:sec5}

Solving a constraint satisfaction problem is in general NP-hard. Thus, this 
section is devoted to instantiating each component of the developed theory 
into concrete predicate languages that can be used in standard solver tools to 
find solutions for the relaxed attack-tree decoration problem. In 
Section~\ref{sec:empirical} below we show how those instantiations of the 
theory can be used to analyse a comprehensive attack tree case study. 

\subsection{Maximal weakening over inequality relations}

Here we address the question of whether there exists a meaningful predicate 
language and constraint satisfaction solver that can be used to solve the 
relaxed attack tree decoration problem with respect to the maximal weakening 
relation. Note that, among all possible weakening relations the maximal one is 
the less restrictive. Hence it leads to more fine-grained solutions than other 
weakening relations. 

The chosen predicate language defines predicates of three types,
all based on comparing one or two labels to a constant value.
The three types of predicates are:
\begin{enumerate}
	\item $\ell\le a$
	\item $\ell\ge a$
	\item $\ell\le\ell' + a$
\end{enumerate}
where $\ell$ and  $\ell'$ are labels and $a$ is a real number (positive or 
negative). We denote the set of all such predicates by $\ineqpred$ and
we will often use $\predType{a}\in\ineqpred$ to denote a predicate in this set
with constant value $a$.
It is easy to verify that $\{p\}\maxweak \{p'\}$ can only hold if
predicates $p$ and $p'$ are of the same type.

\begin{lemma}
Let $p,p'\in\ineqpred$ be two predicates, then
$\{p\}\maxweak \{p'\}$ implies that, for some labels $\ell_1$,
$\ell_2$, and some real numbers $a$, $a'$,
\[
\begin{array}{lll}
p\equiv \ell_1\le a &\land& p'\equiv \ell_1\le a' \text{,~or}\\
p\equiv \ell_1\ge a &\land& p'\equiv \ell_1\ge a' \text{,~or}\\
p\equiv \ell_1\le\ell_2 + a &\land& p'\equiv \ell_1\le\ell_2 + a'\text{.}\\
\end{array}
\]
\end{lemma}

The three types of predicates have been chosen in such a way that the
maximal weakening relation $\maxweak$ on single predicates can be 
easily characterised by the numerical order of their constant values $a$.
This characterisation will allow us later to define the distance
between two predicates as the difference between their constant values
$a$ and $a'$.

\begin{lemma}\label{lem-pred-rel}
Let $\ell_1$, $\ell_2$ be labels and $a$, $a'$ be real numbers. Then
the following properties hold.
\[
\begin{array}{cll}
\{\ell_1\le a\} \maxweak \{\ell_1\le a'\} &\iff& a\ge a'\\
\{\ell_1\ge a\} \maxweak \{\ell_1\ge a'\} &\iff& a\le a'\\
\{\ell_1\le\ell_2 + a\} \maxweak \{\ell_1\le\ell_2 + a'\} &\iff& a\ge a'\\
\end{array}
\]
\end{lemma}

From this lemma it follows, for instance, that for $S,S'\subseteq\ineqpred$,
if $S\maxweak S'$ and $(\ell_1\le a)\in S$, then there exists
$(\ell_1\le a')\in S'$, such that $a\ge a'$. Hence we consider the set 
$\maps{S}{S'}$ containing all \emph{total} 
functions 
$f\colon S \rightarrow S'$ such that 
$\forall p 
\in S \colon  f(p) \implies p$. 
The Euclidean distance between two predicate $\predType{a}, 
\predType{a'}\in\ineqpred$ is given by $d(\predType{a}, \predType{a'}) = 
|a-a'|$ if $\predType{a}$ and $\predType{a'}$ are of the same type, 
$d(\predType{a}, \predType{a'}) = \infty$ otherwise. Given $f \in 
\maps{S}{S'}$, we define $\distance_f(S, S') = \sqrt{\sum_{p \in S} d^2(p, 
f(p))}$, and the Euclidean 
distance $d(S, S')$ between two sets of predicates by

\[
d(S,S') = 
	\left\{ 
  		\begin{array}{l l}
			\infty\text{,} & \text{if } \maps{S}{S'} = \emptyset \\
			\min_{f \in \maps{S}{S'}} \distance_f(S, S') \text{,} & 
			\text{otherwise}
  		\end{array}
	\right.
\]

We restrict the distance measure above to bijective functions only. That is to 
say, we consider from now on $\maps{S}{S'}$ to be the set containing all 
\emph{bijective} functions $f\colon S \rightarrow S'$ such that 
$\forall p \in S \colon  f(p) \implies p$. A consequence of such restriction is 
that predicate sets with different cardinality have distance $\infty$, which 
simplifies the proof of Theorem~\ref{theo-main} below. 

Next we provide sufficient conditions for a set of predicates to be part of a 
solution of the relaxed attack-tree decoration problem. It states that a set 
$S$, which minimizes its distance $\metric(S, S(t))$ to $S(t)$, where $S(t)$ is 
the set of soft 
predicates for a given tree $\anatree$, and that satisfies $\thesem{H(t) \cup 
S} \neq \emptyset$, where $H(t)$ is the set of hard predicates, leads to a 
solution of the relaxed attack-tree decoration problem.

\begin{theorem}\label{theo-main}
 Let $\anatree$ be an attack tree and $H(t) \cup S(t)$ an 
 attribute constraint-set for $\anatree$, where $S(t)\subseteq\ineqpred$.
 Let $S\subseteq\ineqpred$ be a set of predicates such that 
 $\thesem{H(t) \cup S} \neq \emptyset$ and 
$\metric(S, 
 S(t))$ is minimum and defined, i.e.\ $\metric(S, 
 S(t)) \neq \infty$. Then there exists a valuation $\alpha \in \thesem{H(t) 
 \cup 
 S}$ such 
 that $(S, 
 \alpha)$ is a solution of the relaxed attack-tree 
 decoration problem with respect to $\maxweak$. 
\end{theorem}
\begin{proof}

First, if $\thesem{H(t) \cup S(t)} \neq \emptyset$, then there exists $\alpha 
\in \thesem{H(t) \cup S(t)}$ and $\metric(S(t), S(t)) = 0$, which is minimum. 
Thus in the remainder of the proof we assume that $\thesem{H(t) \cup S(t)} = 
\emptyset$.

Now, notice that $S(t) \not \maxweak S$, otherwise $\thesem{H(t) \cup S} = 
\emptyset$. Thus we obtain that $S \maxweak S(t)$, implying that $S$ satisfies 
the first  
condition of the relaxed attack tree decoration problem. Moreover, we also 
conclude that 
$|S| = |S(t)|$, given that $\metric(S, S(t)) \neq \infty$. 
Next, we will 
show that $S$ satisfies the third condition 
of the problem definition as well. 

Suppose we have 
$S \maxweak S' \maxweak S(t)$, but 
$S' \not \maxweak S$. 
Because 
$S' \not \maxweak S$, there must exist  $p' \in 
S'$ such that 
no $p \in S$ satisfies that $p \implies p'$. Let us analyse the three possible 
predicate types of $p'$.

\begin{enumerate}
	\item Assume $p' \equiv \ell \le a'$. Because $S' \maxweak S(t)$, there 
	must exist 
$p'' \equiv \ell \leq a''$ in $S(t)$ with $a'' \leq a'$. Now, let 
$f: S \rightarrow S(t)$ be a bijective function such that $\distance_f(S, S(t)) 
= 
\metric(S, 
S(t))$. Such a function exists given 
that $\metric(S, S(t)) 
\neq \infty$ and $S \maxweak S(t)$. Let $p \equiv \ell \leq 
a$ be the predicate in $S$ such that $f(p) = p''$. Overall we obtain that $a'' 
\leq a'$ and $a'' \leq a$. Now, given that $p$ does not imply $p'$, according 
to Lemma~\ref{lem-pred-rel} 
it must be the case that $a > a'$. Therefore we obtain the order $a > a' \geq 
a''$.  Consider the set of predicates $S'' = S \setminus \{p\} \cup \{p'\}$. 
	On the 
one hand, because 
$a' \geq a''$ it follows that $S'' \maxweak S(t)$ and $f \in \maps{S''}{S(t)}$. 
On other hand, 
because $a > a'$ and $a' \geq a''$, we obtain that $\distance_f(S, S(t)) > 
\distance_f(S'', 
S(t))$. Considering that $\distance_f(S, S(t)) =\metric(S, S(t))$, then 
$\metric(S, 
S(t)) > 
\distance_f(S'', S(t)) \geq \metric(S'', S(t))$, which contradicts the 
assumption 
that $\metric(S, S(t))$ is minimum. 

	\item The case $p' \equiv \ell = \ell' \leq a'$ is analogous to the 
previous one.  
	\item Finally assume $p' \equiv \ell \geq a'$. As in the first case we 
	obtain that 
there exist predicates $p \equiv \ell \geq 
a$ and $p'' \equiv \ell \geq a''$ in $S$ and $S(t)$, respectively, such that 
$p'' \implies p'$ and $f(p) = p''$. Given that $p$ does not imply $p'$, then we 
obtain the following order, $a'' \geq a' > a$. Again, such an order implies 
that $\distance(p, p'') > \distance(p', p'')$. Therefore, the set of predicates 
$S'' = S 
\setminus \{p\} \cup \{p'\}$ satisfies that $\metric(S, S(t)) > \metric(S'', 
S(t))$, 
which contradicts the assumption that 
$\metric(S, S(t))$ is minimum.  
\end{enumerate}
The proof concludes by remarking that $S$ satisfies the second condition of the 
relaxed attack tree decoration problem, as stated in the body of the theorem. 
\end{proof}

This theorem demonstrates a reduction of the 
relaxed attack-tree decoration problem for the maximal weakening 
relation $\maxweak$ on the given types of predicates to an 
optimisation problem that we solve via nonlinear programming. The formulation 
of the optimisation problem is given below.

\begin{corollary}
 Let $\anatree$ be an attack tree and $H(t) \cup S(t)$ an 
 attribute constraint-set for $\anatree$ where $S(t) = \{p_1, \ldots, p_n\}$ 
 contains only predicates from $\ineqpred$.
 We create the set of predicates 
$S$ by replacing every predicate $p_i \in S(t)$ by.
\[
\begin{array}{ll}
p_{\alpha} \equiv \ell \le x_i & \mbox{if $p \equiv 
\ell\le a$}\\
p_{\alpha} \equiv \ell\ge x_i  & \mbox{if $p \equiv 
\ell\ge a$}\\
p_{\alpha} \equiv \ell\le \ell' + x_i & 
\mbox{if $p \equiv 
\ell\le \ell' + a$}
\end{array}
\]  
where $x_1, \ldots, x_i$ are variables. 
A solution $\alpha, x_1, \ldots, x_n$ to the following optimisation problem 
leads to a solution 
$(S, \alpha)$ of the relaxed attack tree problem.
\begin{equation*}
\begin{aligned}
& \underset{\alpha, x_1, \ldots, x_n}{\text{minimize}}
& & d(S, S(t)) \\
& \text{subject to}
& & \alpha \in \thesem{H(t) \cup S}.
\end{aligned}
\end{equation*}
\end{corollary}

\noindent\emph{Implementation.} 
To find a valuation function $\valuationFunc$ and a set of weakening predicates 
$S'$ so that the distance function $d(S,S(t))$ 
is minimised, we relied on the Sequential Quadratic Programming (SQP) problem 
interpretation of the relaxed attack-tree decoration problem. We further refer 
to this tool as the SQP-based tool. It is implemented using the Python 
\texttt{scipy} library that provides the Sequential Least Squares Programming 
(SLSQP) algorithm for solving optimisation problems of this type. Our 
implementation\footnote{Code available at 
\url{https://github.com/vilena/at-decorator/tree/master/SQP_decorator}} 
does not impose any 
burden on the analyst, as it allows a loose interpretation of all 
constraints together. In case the set of constraints is satisfiable, our tool 
finds an optimal solution. In case of an unsatisfiable set of constraints, our 
implementation will find an optimal solution that minimises the distance 
function between set of predicates.

\subsection{Set inclusion weakening over propositional logic}

As the analyst may require a predicate language richer 
than the one described above, we provide tool support for predicates written in 
the
propositional logic. We do so via a transformation of a relaxed attack 
tree-decoration problem instance into a Satisfiability Modulo Theories (SMT) 
instance~\cite{MouraB08}. 
SMT is the problem of determining whether a formula in the first-order logic, 
where 
some operator symbols are provided with a theory, is satisfiable. 

Given a set of predicates $S$, we use $\Gamma(S)$ to denote the first-order 
logic formula formed by all predicates in $S$ in the conjunctive form. 
Then the SMT instance resulting from the relaxed attack-tree decoration problem 
instance
with 
the attribute constraint-set $H(t) \cup S(t)$ is defined by $\Gamma(H(t) \cup 
S(t))$. If  $\Gamma(H(t) \cup 
S(t))$ is satisfiable, it follows that the decoration problem does not need to 
be relaxed. 
Otherwise, we use 
Algorithm~\ref{alg-inclusion-relation} to find a 
subset of soft predicates $S 
\subset S(t)$ that solves the relaxed attack tree decoration problem with 
respect to the inclusion relation. 

\floatname{algorithm}{Algorithm}
\begin{algorithm}[t!]
  \caption{Solving the decoration problem w.r.t. the set 
  inclusion weakening relation. \label{alg-inclusion-relation}}
  \begin{algorithmic}[1]
  \Require The relaxed attack-tree decoration problem defined by attack tree 
  $t$, attribute constraint-set $H(t) \cup S(t)$.
  \Ensure $S$ is a set of maximum cardinality such that $S \subseteq S(t)$ and 
  $\thesem{H(t) \cup S}\neq \emptyset$ .
  \State Let $S(t)$ = $\{p_1, \dots, p_n\}$ be the soft constraint set.

      \State $S \gets \emptyset$
      \For {$i = 1 .. n$}
          
          \State $S \gets S \cup \{p_i\}$
          \If {$\Gamma(H(t) \cup S)$ is not satisfiable}
              \State $S \gets S \backslash \{p_i\}$
          \EndIf 
      \EndFor
      
      \State \Return $(S, \alpha)$ where $\alpha \in \thesem{H(t) \cup 
      S}$.

  \end{algorithmic}
\end{algorithm}

Initially, the set $S$ is empty. Algorithm~\ref{alg-inclusion-relation} works 
by iteratively adding predicates 
to the set $S$, until the formula $\Gamma(H(t) \cup 
S)$ is satisfiable, while the formula $\Gamma(H(t) \cup 
S')$ is unsatisfiable for any $S' \supset S$. 
It is easy to prove that such procedure provides a 
solution to the relaxed 
attack-tree decoration problem with respect to the subset inclusion weakening 
relation, based on the assumption that hard predicates are satisfiable.

\noindent\emph{Implementation.}
We implemented our transformation relying on the well-known theorem prover 
Z3 from Microsoft\footnote{\url{https://github.com/Z3Prover/z3}}. Z3 can be 
utilised as 
a constraint solver, i.e., it can find a solution satisfying a set  of 
constraints expressed as equalities and inequalities. 
Our current implementation\footnote{Code available at 
\url{https://github.com/vilena/at-decorator/tree/master/CSP_decorator}} 
can handle all attribute domains on 
real
numbers that are defined in the ADTool~\cite{KordyKMS-QEST-13},
e.g.\ probability, minimal cost of attack, and minimal time; and it is
trivially extensible to attribute domains defined on Boolean values,
e.g.\ satisfiability of a scenario. The
analyst can further specify additional constraints, if desired. The
resulting set of constraints is passed to the Z3 prover, which reports
whether the problem is solvable or not. In case the constraint satisfaction
problem is solvable, the prover will report a complete consistent
valuation for the given tree. This valuation will satisfy the tree
structure and the constraints expressed by the analyst, and it will
agree with the given initial valuation. If the constraint satisfaction
problem is not solvable (i.e.\ the initial valuation is inconsistent),
we use Algorithm~\ref{alg-inclusion-relation} to find a subset of soft 
predicates of maximum cardinality that is satisfiable, and report a 
solution satisfying the found maximal predicate set.


\section{ATM Case Study}
\label{sec-evaluation}

We now proceed to show how our approach to attack-tree decoration can be 
applied to a real-life security scenario. The goal is to show that decoration 
can be performed in a systematic way even when the analyst only has partial 
information 
about attribute values. 

In this section, we introduce a case study related to capturing automated 
teller machine fraud
scenarios as an attack tree, and we describe the historical data
available for decorating this tree.

\subsection{ATM Security: a Case Study}

Automated teller machines (ATMs) are complex and expensive 
systems used daily by millions of bank customers worldwide. Because 
each carries a significant amount of cash, ATMs are the target of large-scale 
criminal actions. Only in 2015 
more than $16,000$ ATM incidents were reported in Europe, causing over $300$ 
million Euros
loss\footnote{\url{https://www.association-secure-transactions.eu/tag/atm-crime-report/}}.

In an attempt to provide structure to the risk assessment process and catalogue 
ATM threats, Fraile et al.\ created a comprehensive attack-defence tree 
capturing the most dangerous attack scenarios applicable to 
ATMs~\cite{PoEM-2016}. The tree modelled in~\cite{PoEM-2016} contains three 
main branches:
brute-force attacks, fraud 
attacks, and logical attacks. The attacks of the logical type make use of 
malicious software, while a brute-force attack typically ends up destroying the 
ATM.
Differently from these two attack scenarios, ATM fraud attacks involve 
conventional 
electronic devices (such as card skimmers) and require the 
participation of the victim.

\begin{figure*}[!ht]
	\centering
	\includegraphics[width=0.88\textwidth]{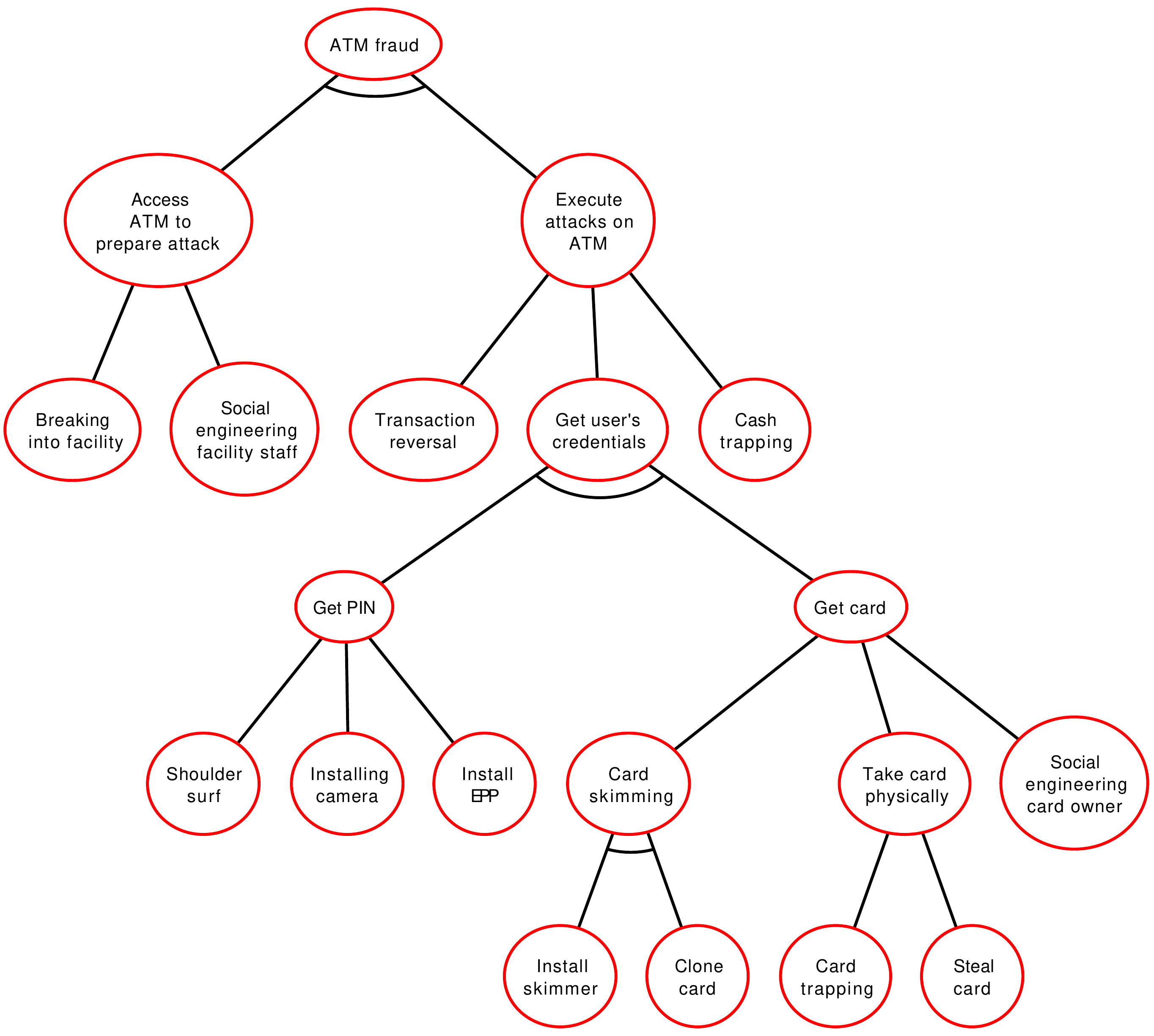}
	\caption{An attack tree modelling ATM fraud. The tree is loosely based on 
	the 
		attack-defence tree published by Fraile et al.~\cite{PoEM-2016}. 
		\label{fig-tree}}
\end{figure*}

In this empirical evaluation section we focus on ATM 
fraud, because more empirical data is available for 
these attacks than for the other types of attacks.
Figure~\ref{fig-tree} presents an attack tree characterising such attacks that 
is loosely based on 
	the attack-defence tree published by Fraile et al.~\cite{PoEM-2016}. In ATM 
	fraud, criminals need covert access to the 
ATM, as this attack typically 
requires 
opening the machine's case either by force or with a generic 
key, and installing a special device  (e.g.\ a skimmer). Then the attacker 
waits until a victim uses the ATM and, as a 
consequence, enables the installed device. Lastly, the 
attacker gets cash from the victim's account by means of various techniques, 
such as 
cash trapping, card cloning, etc.

\subsection{Decorating the ATM Fraud Attack Tree}

The decoration process we propose in this paper consists of three independent 
steps that are executed for a given attack tree. First, an attribute is chosen. 
In this case study, we focus on 
\emph{probability of success}, that is, the probability that a 
given ATM machine is used to successfully execute ATM fraud. The
attack tree structure jointly with the attribute rules determine the
hard constraint set, i.e.\ the standard bottom-up constraints derived from the 
attack thee structure.
Second, statistical information (historical data) related to the chosen 
attribute is 
gathered. Such statistical values are used to provide tree nodes with 
probability values. For the ATM fraud scenario, the available statistical data 
is presented in Section~\ref{sec:statistical-analysis}.
Lastly, relations among nodes in the tree are established based on the 
analyst's insight and domain knowledge. Section~\ref{sec:hard-constraints} 
presents the corresponding analysis for the ATM fraud scenario. The full set of 
constraints for the ATM fraud scenario is summarised in 
Section~\ref{sec:predicate-set}.

\subsubsection{Statistical Analysis}\label{sec:statistical-analysis}
The statistical values we consider here 
have been derived from the ATM Crime Report 2015 (EAST). In our case, we 
analyse 
ATM fraud incidents in Lisbon, which hosts 300 ATMs. We remark, however, 
that these values have been derived for illustrative purposes only and may not 
be accurate. 

Between $2010$ and $2015$, $83$ ATM fraud attacks have been reported in Lisbon. 
This gives a $0.0461$ probability of an ATM to be the target of fraud
within a calendar year, if we assume the uniform distribution of these
attacks. Because the report categorises ATM fraud into different
attack 
types, we can provide probability values for some attack types by analysing the 
attack frequency as reported in the EAST report. The results can be found in 
Table~\ref{table-stast-prob}. For the results reported in this table, we assume 
that historical attacks were uniformly distributed, and we rely on frequencies 
of attacks over long time periods to estimate probabilities, like in the OCTAVE 
method \cite{OCTAVE-2007}.

\begin{table}
	\scriptsize
	\centering
	\caption{\label{table-stast-prob}Historical data values identified for some 
	attack tree nodes from the ATM Crime report.}
	\begin{tabular}{m{2.0cm} m{0.1cm} m{1cm} m{0.1cm} m{3.5cm}}
		\toprule
		\bfseries Node && 
		\bfseries Prob. && 
		\bfseries Source \\
		\midrule
		\ATMFraudLabel  && $0.0046$ && ATM Crime Report 2015 (EAST) \\
		\midrule
		\CardSkimmingLabel && $0.0172$ && ATM Crime Report 2015 (EAST) \\
		\midrule
		\CardTrappingLabel && $0.0094$ && ATM Crime Report 2015 (EAST) \\
		\midrule
		\CashTrappingLabel && $0.0150$ && ATM Crime Report 2015 (EAST) \\
		\midrule
		\TransactionReversalLabel  && $0.0038$ && ATM Crime Report 2015 (EAST) 
		\\
		\bottomrule
	\end{tabular}

\end{table}

\subsubsection{Domain Knowledge Constraints}\label{sec:hard-constraints}

In the previous subsection, we have shown how available statistical data  can 
be used as a constraint in our decoration process. 
Another novelty of our approach 
is that we allow for domain knowledge constraints, that is, facts that must be 
additionally satisfied in 
the attack tree. The following list of facts is based on the previously 
mentioned ATM Crime Report 2015 and also on the European Central Bank report on 
card fraud (2015). 

\begin{itemize}
	\item \CardSkimmingLabel\ is more likely than \TakeCardLabel. 
	Moreover, \GetCredentialsLabel\ is more likely than \CashTrappingLabel, 
	which is more likely than \TransactionReversalLabel.
	\item \ShoulderSurfLabel\ is more likely than \InstallCameraLabel.
	\item  \InstallCameraLabel, \InstallEPPLabel\ and \InstallSkimmerLabel\ 
	are all equally likely.
	\item \CashTrappingLabel\ and \CardTrappingLabel\ are equally likely.
\end{itemize}

\subsection{Full Set of Predicates}\label{sec:predicate-set}
We now list the full set of predicates that will be used by the tools to solve 
the decoration problem. To simplify the presentation, we use a short-hand 
notation. All predicates listed in this section can be straightforwardly 
transformed into our predicate notation.

\noindent\emph{Hard predicates.}
Considering the attribute domain of probability of success, the attack tree 
shown in Fig.~\ref{fig-tree} corresponds to the following set of hard 
predicates:

\begin{itemize}
	\item \ATMFraudLabel\ = $\times$(\AccessATMLabel, \ExecuteAttackLabel),
	\item \AccessATMLabel\ = $+$(\BreakingIntoLabel, 
	\SocialEngineeringStaffLabel),
	\item \ExecuteAttackLabel\ = $+$(\TransactionReversalLabel, 
	\GetCredentialsLabel, \CashTrappingLabel),
	\item \GetCredentialsLabel\ = $\times$(\GetPINLabel, \GetCardLabel),
	\item \GetPINLabel\ = $+$(\ShoulderSurfLabel, \InstallCameraLabel, 
	\InstallEPPLabel),
	\item \GetCardLabel\ = $+$(\CardSkimmingLabel, \TakeCardLabel, 
	\SocialEngineerOwnerLabel),
	\item \CardSkimmingLabel\ = $\times$(\InstallSkimmerLabel, \CloneCardLabel),
	\item \TakeCardLabel\ = $+$(\CardTrappingLabel, \StealCardLabel).
\end{itemize}

\noindent\emph{Soft predicates.}
Historical data values from Table~\ref{table-stast-prob} are encoded in the 
form of soft predicates:
\begin{itemize}
	\item \ATMFraudLabel\ = $0.0046$, 
	\item \CardSkimmingLabel\ = $0.0172$,
	\item \CardTrappingLabel\ = $0.0094$,
	\item \CashTrappingLabel\ = $0.0150$,
	\item \TransactionReversalLabel\ = $0.0038$. 
\end{itemize}

We will subsequently refer to the soft predicates listed above as 
\emph{historical data constraints}.

Domain knowledge from the ATM Crime report is encoded in the form of soft 
predicates as well:
\begin{itemize}
	\item \TakeCardLabel\ $\leqslant$ \CardSkimmingLabel,
	\item \CashTrappingLabel\ $\leqslant$ \GetCredentialsLabel,
	\item \TransactionReversalLabel\ $\leqslant$ \CashTrappingLabel,
	\item \InstallCameraLabel\ $\leqslant$ \ShoulderSurfLabel,
	\item \InstallCameraLabel\ = \InstallEPPLabel,
	\item \InstallSkimmerLabel\ = \InstallEPPLabel,
	\item \InstallSkimmerLabel\ = \InstallCameraLabel,
	\item \CashTrappingLabel\ = \CardTrappingLabel,
\end{itemize}
Subsequently, we will refer to the set of soft predicates above as \emph{domain 
knowledge predicates}.

\subsection{Goals of the Analysis}

We consider that the analyst has designed an attack tree covering ATM fraud 
scenarios as presented in Figure~\ref{fig-tree}. This attack tree gives them 
the set of hard constraints given in Section~\ref{sec:predicate-set}. 
Furthermore, the analyst has elicited a set of soft constraints based on their 
knowledge of the problem space and the information available in the ATM Crime 
Report (also listed in Section~\ref{sec:predicate-set}). However, the analyst 
is not able to find enough data to estimate probabilities for all leaf nodes in 
the attack tree, what prevents them from straightforwardly applying the 
bottom-up evaluation procedure to compute the probabilities for all 
intermediate nodes and, ultimately, for the root node.

The analyst can, however, apply our methodology and decorate the attack tree. 
We consider the following possible analysis questions that can be investigated 
with our approach:
\begin{itemize}
	\item Are the attack tree and the set of constraints elicited by the
  analyst compatible? I.e.\ does the corresponding decoration problem have a 
  solution? In our notation, for an attack tree $\anatree$, and attribute 
  constraint-set $H(t) \cup S(t)$, is $ \thesem{H(t) \cup S(t)} \neq \emptyset 
  $?

	\item What is a solution for the given decoration problem? In our notation, 
	the analyst is interested in finding a solution $\alpha \in \thesem{H(t) 
	\cup S(t)}$.

	\item If the decoration problem has no solution, what is a solution that is 
	the closest to satisfying all constraints? This question corresponds to 
	solving the relaxed attack tree decoration problem formulated in 
	Definition~\ref{def:relaxed_decoration_problem}, for a chosen weakening 
	relation. 
\end{itemize}
We will demonstrate in the next section how our two implementations solve the 
relaxed attack tree decoration problem of the ATM case study, for the maximal 
and set inclusion weakening relations.

Note that one may argue that in our case study the analyst already has the 
probability of ATM fraud (the root node) from the historical data, and 
therefore, they can skip decorating the whole tree. 
However, the analyst may still want to perform the so-called \emph{what-if} 
analysis,  
which consists in analyzing different but related scenarios. For example, the 
analyst could answer questions such as: What if the 
probability of this attack is in fact higher than I envisage? How will this 
affect my security posture? The 
what-if analysis requires a fully annotated tree, which can be provided using 
our decoration technique even over partially available data. 
Furthermore, in 
general, it cannot be assumed that the root node value will always be available 
from the historical data.

\section{Empirical Evaluation Results}\label{sec:empirical}

We now show how our two implementations can be applied to analyse the
ATM case study introduced previously.

\begin{table*}[!t]
\centering
\caption{Solutions of the ATM fraud attack tree decoration problem found by our 
tools}
\label{tbl:tool-results}
\footnotesize
\begin{tabular}{r|c|c|c}\hline\hline
\multirow{5}{*}{Node Label} & \multicolumn{3}{c}{Probability of attack success 
for tool and set of predicates}\\ \cline{2-4}
& CSP & CSP & SQP \\
& hard predicates + & hard predicates + & hard predicates + \\
& historical data & historical data + & historical data + \\
& & domain knowledge & domain knowledge \\
\hline\hline
\ATMFraudLabel & $0.0046$ & $0.0046$ & $0.0046$\\ 
\AccessATMLabel & $0.0068$ & $0.0093$ & $0.0184$\\
\BreakingIntoLabel & $0.0039$ & $0.0078$ & $0.0092$\\
\SocialEngineeringStaffLabel & $0.0029$ & $0.0015$ &  $0.0092$\\
\ExecuteAttackLabel & $0.6683$ & $0.4914$ &  $0.2493$\\
\TransactionReversalLabel & $0.0038$ & $0.0038$ & $0.0038$\\
\GetCredentialsLabel & $0.6620$ & $0.4847$ & $0.2324$\\
\GetPINLabel & $0.875$ & $0.9375$ & $0.5780$\\
\ShoulderSurfLabel & $0.5$ & $0.75$ & $0.3834$\\
\InstallCameraLabel & $0.5$ & $0.5$ & $0.0973$\\
\InstallEPPLabel & $0.5$ & $0.5$ & $0.0973$\\
\GetCardLabel & $0.7566$ & $0.5170$ & $0.4021$\\
\CardSkimmingLabel & $0.0172$ & $0.0172$ & $0.0172$\\
\InstallSkimmerLabel & $0.5$ & $0.5$ & $0.0973$\\
\CloneCardLabel & $0.0344$ & $0.0344$ & $0.1768$\\
\TakeCardLabel & $0.5047$ & $0.0171$ & $0.0172$\\
\CardTrappingLabel & $0.0094$ & $0.0094$ & $0.0113$\\
\StealCardLabel & $0.5$ & $0.0078$ & $0.0059$\\
\SocialEngineerOwnerLabel & $0.5$ & $0.5$ &  $0.3677$\\
\CashTrappingLabel & $0.015$ & $0.0094$ & $0.0131$\\
\hline
\end{tabular}
\end{table*}

\begin{figure*}[!ht]
	\centering
	\includegraphics[width=1\textwidth]{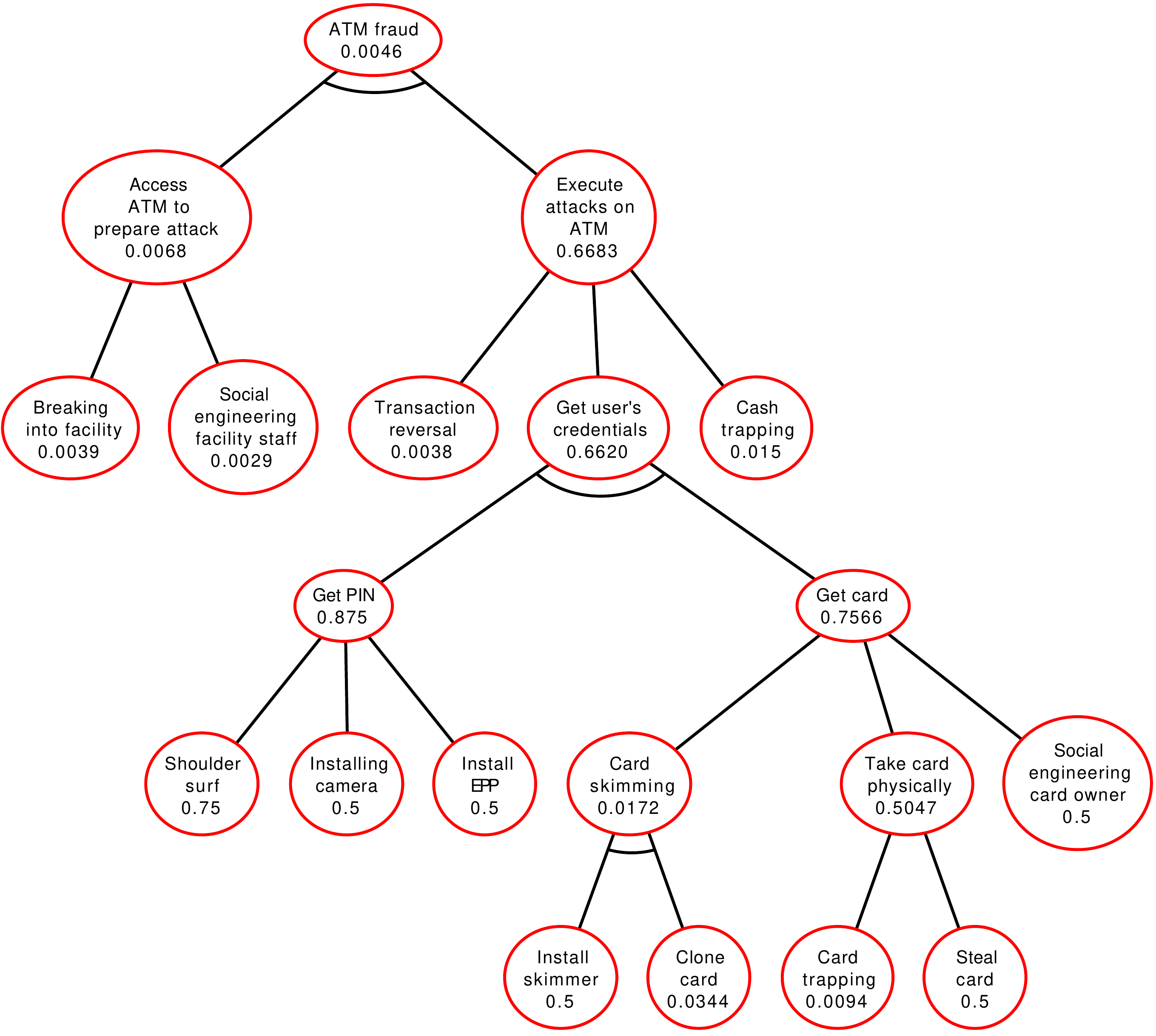}
	\caption{A valuation identified by the CSP-based tool for the attack tree 
	in 
	Fig.~\ref{fig-tree} with the hard constraints and the historical data 
	predicates used as soft constraints (predicates are listed in 
	Section~\ref{sec:predicate-set}).} 
	\label{fig-tree-before}
\end{figure*}

\begin{figure*}[!ht]
	\centering
	\includegraphics[width=1\textwidth]{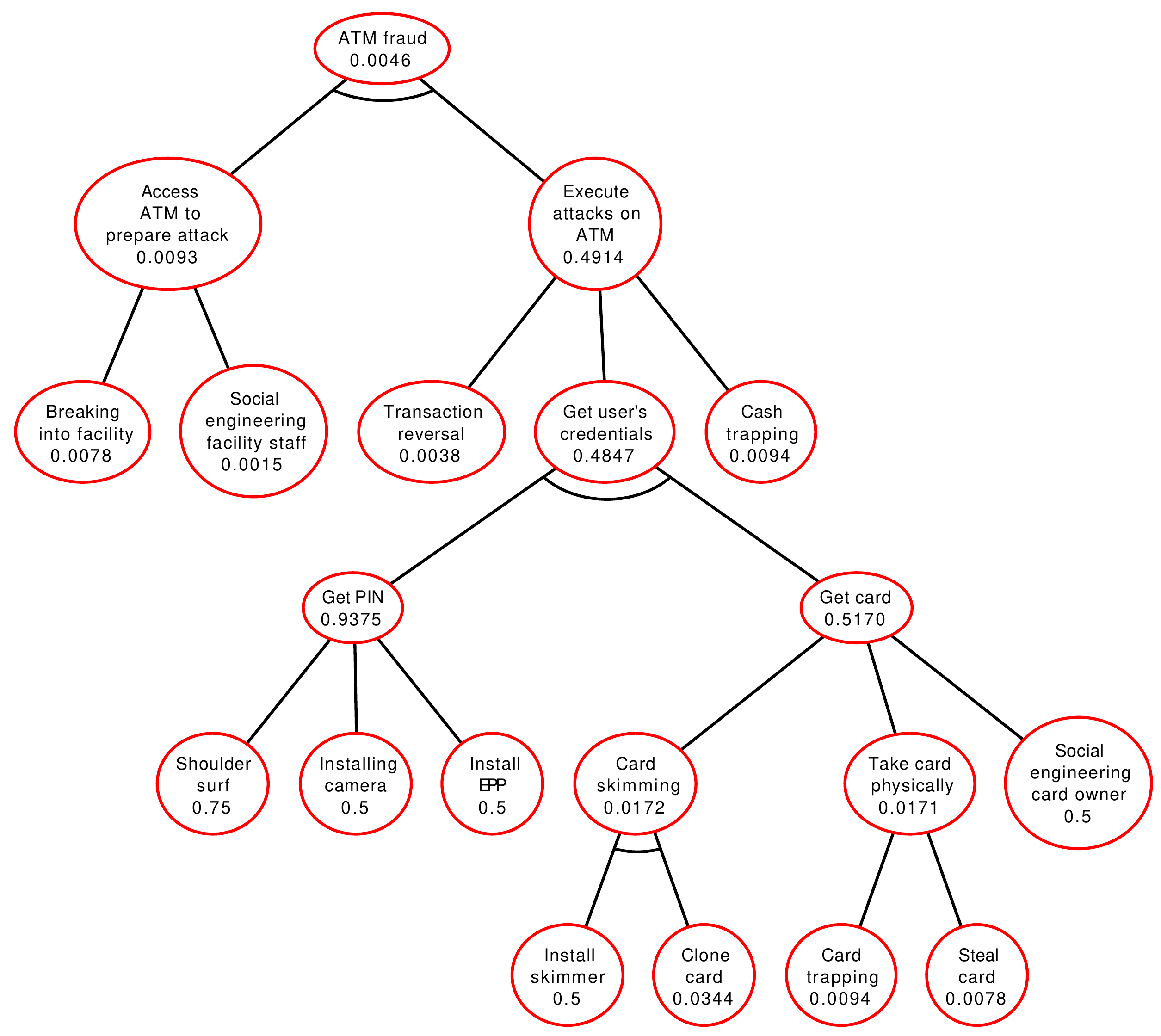}
	\caption{A valuation identified by the CSP-based tool for the attack tree 
	in 
	Fig.~\ref{fig-tree} with the predicates listed in 
	Section~\ref{sec:hard-constraints}.} 
	\label{fig-tree-after}
\end{figure*}

\begin{figure*}[!ht]
	\centering
	\includegraphics[width=1\textwidth]{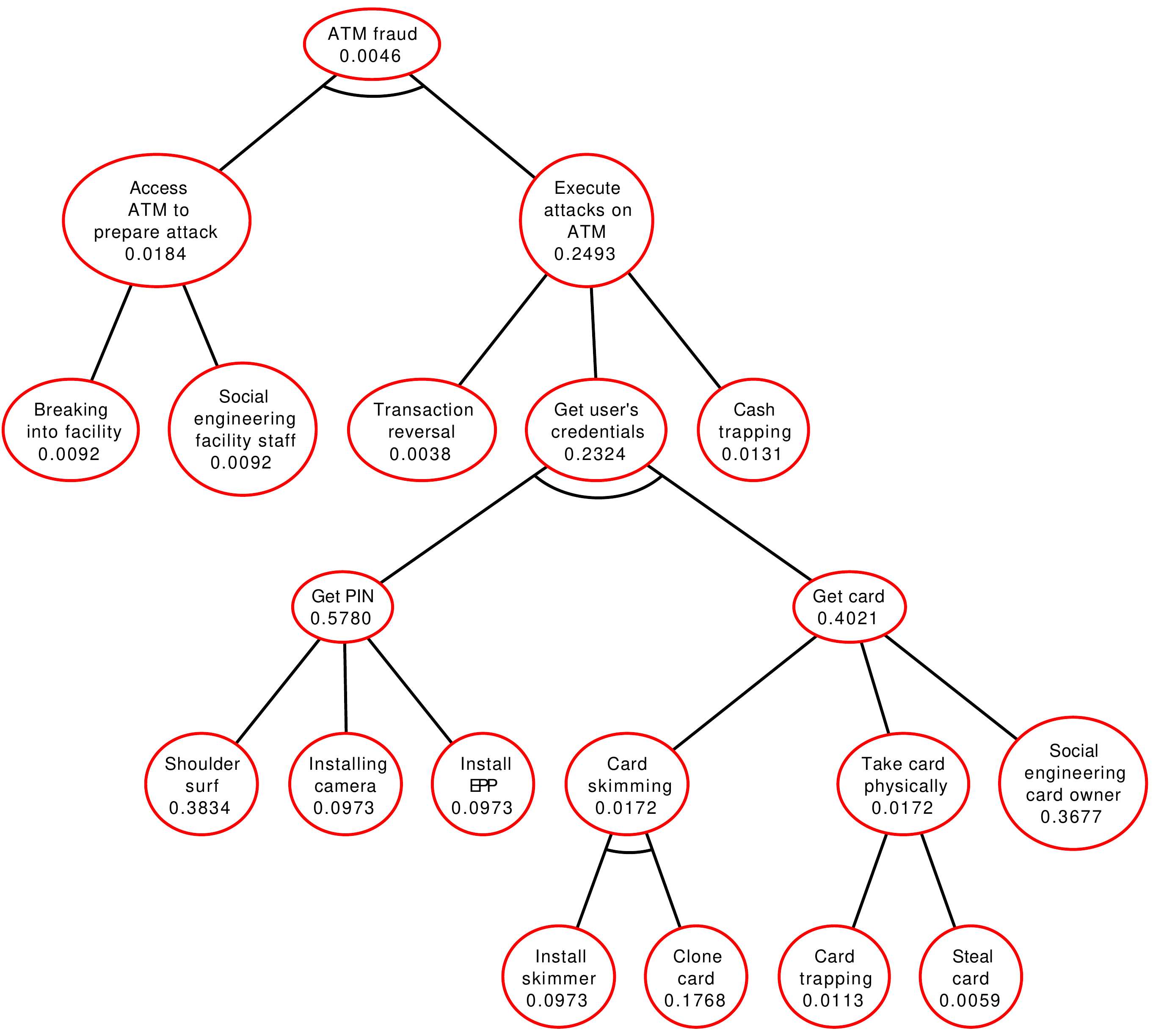}
	\caption{A valuation identified by the SQP-based tool for the attack tree 
	in 
	Fig.~\ref{fig-tree} with all hard and soft predicates listed in 
	Section~\ref{sec:hard-constraints}.} 
	\label{fig-tree-sqp}
\end{figure*}

\subsection{The CSP-based Implementation Showcase}

We first exemplify the results of the CSP-based implementation on the
ATM case study tree presented in Fig.~\ref{fig-tree}. The first
solution in Table~\ref{tbl:tool-results}
presents a possible valuation for this attack tree found by the CSP-based tool 
with the hard predicates and the historical data predicates listed in 
Section~\ref{sec:predicate-set}.  Figure~\ref{fig-tree-before}  visualises this 
solution in the ATM fraud attack tree. We have used the open source ADTool 
software 
\cite{KordyKMS-QEST-13,gadyatskaya2016attack} for visualising attack trees.

If the analyst introduces the domain knowledge predicates as an
additional set of constraints, the CSP tool indicates that
the constraint satisfaction problem becomes
unsatisfiable. Indeed, the constraints on \CardTrappingLabel\ and 
\CashTrappingLabel\ are contradictory, because historical data does not 
indicate that the probabilities of these attacks are exactly equal. 

The Z3 solver that we use as the underlying constraint satisfaction engine is 
capable of finding an unsatisfiable core of the problem. In the ATM fraud case, 
the solver reports that three predicates constitute the unsatisfiable core: 
\begin{itemize}
	\item \CardTrappingLabel\ $=$ \CashTrappingLabel, 
	\item \CardTrappingLabel\ $= 0.0094$,
	\item \CashTrappingLabel\ $= 0.015$.
\end{itemize}

The solver then proceeds to identify a maximal set of predicates that is still 
satisfiable. In our case, the solver is able to satisfy all predicates but 
\CashTrappingLabel\ $= 0.015$. The second solution in 
Table~\ref{tbl:tool-results}, illustrated in
Figure~\ref{fig-tree-after}, shows a possible
valuation found by the CSP-based tool after adding
the set of domain knowledge constraints to the original assignment
predicate set (i.e.\ excluding the constraint on equality of 
\CashTrappingLabel\ to 0.015).
 
Analysing this valuation, we can see that, for example, the probability of 
\TakeCardLabel\ has dropped significantly, and the
probability of \StealCardLabel\ has decreased with the new constraint
set. Indeed, in the first assignment set there was no constraint that
the probability of \TakeCardLabel\ is less than the
probability of \CardSkimmingLabel\ (which is statistically a rare event
itself). When we added this constraint found in the domain knowledge
catalogue, the valuation for this subtree has changed. This result
demonstrates that even if the analyst cannot obtain a valuation for some attack 
such as \StealCardLabel, they can get an estimation for it by
using other known parameters and domain knowledge.

\subsection{The SQP-based Implementation Showcase}

When running the relaxed tree decoration method considering the historical
data values in Table~\ref{table-stast-prob} and the domain knowledge
constraints from Section~\ref{sec:hard-constraints}, the optimisation
task converged to the solution presented in the third solution of
Table~\ref{tbl:tool-results}. Figure~\ref{fig-tree-sqp} presents the
decorated attack tree. 

Three of the soft constraints stemming from the historical data values 
(\ATMFraudLabel, \CardSkimmingLabel, \TransactionReversalLabel) were satisfied 
precisely without weakening. \CashTrappingLabel\ and \CardTrappingLabel\ 
constraints were satisfied by weakened predicates. The probability of 
\CardTrappingLabel\ in historical data was $0.0094$, and the optimised value is 
$0.0113$. The probability of \CashTrappingLabel\ was $0.015$ in historical 
data, and the optimised value is $0.0131$.

Six out of seven domain knowledge constraints were satisfied precisely without 
weakening.
The remaining constraint that was not satisfied precisely is the knowledge that 
\CashTrappingLabel\ and \CardTrappingLabel\ are equally likely. In the 
optimised solution they have a small difference, the probability of 
\CashTrappingLabel\ is $0.0131$ while the probability of \CardTrappingLabel\ is 
$0.0113$.

In the result of this process, four weakened predicates were found.
The soft and weakened predicates for these four cases are shown in
Table~\ref{tbl:fwp}. It is the task of the analysts to decide on whether such 
weakened predicates fit their analysis. 
For example, except for the $\TakeCardLabel$ predicate, all other 
predicates in Table~\ref{tbl:fwp} satisfy that their distance to their 
weakening predicate is lower 
than $10^{-2}$, and one may argue that in the 
probability domain a $10^{-2}$ discrepancy is acceptable.

\begin{table*}[!ht]
\caption{Weakened predicates found by the optimisation process}
\label{tbl:fwp}
\footnotesize
\centering
\begin{tabular}{ll}\hline
Soft Predicate & Weakened Predicate\\\hline\hline
\CardTrappingLabel\ $\leqslant 0.0094$ & \CardTrappingLabel\ $\leqslant 
0.0113$\\
\CashTrappingLabel\ $\leqslant 0.015$ & \CashTrappingLabel\ $\leqslant 0.0131$\\
\TakeCardLabel\ $\leqslant$ \CardSkimmingLabel\ + $0$ & \TakeCardLabel\ 
$\leqslant$ \CardSkimmingLabel\ + $0.0172$\\
\CashTrappingLabel\ $\leqslant$ \CardTrappingLabel\ + $0$ & \CashTrappingLabel\ 
$\leqslant$ \CardTrappingLabel\ + $0.0018$\\\hline
\end{tabular}
\end{table*}

\subsection{Discussion}\label{sec:discussion}
We have showcased the application of our methodology on a practical
scenario of ATM security. We have established in this empirical
validation that the attack-tree decoration methodology is versatile,
as it can be realised differently depending on the analyst's needs.
In particular, the analyst can choose a suitable weakening relation and ask
the solver to tackle the decoration problem depending on their
confidence in the tree structure and the historical data available.

When the confidence in the predicate set is high, the analyst will
work with the relaxed attack-tree decoration problem, e.g.\ using our
CSP and SQP-based tools, to obtain a solution that is the closest one to
satisfying all available data points. When the analyst has low
confidence in the data, they can work with a constraint solver, for
example, starting from our CSP-based implementation, to identify whether there
is inconsistency in the data. If the constraint satisfaction problem
is not satisfiable, the reported unsatisfiable core of the predicates
can be the first candidate to be reviewed with other experts in order to
 revise the corresponding data points.

The limitation of our methodology is that it does not allow to draw precise 
conclusions about data accuracy. For example, an unsatisfiable core reported by 
the CSP-based implementation in case of inconsistency does not guarantee that 
this is indeed the set of wrong predicates. It might be that the real problem 
lies within another, co-dependent set of predicates, where considered data 
values are inaccurate.
We plan to investigate the means to evaluate the decoration accuracy
considering e.g.\ analyst's confidence level for each data point and the size 
of the solution space in future work.

Our methodology does not fully remove the intellectual
burden from the analyst, but it equips them with an insight into
conformity of the historical data with the given attack tree
structure, and with an initial, complete attack-tree decoration. This
decoration can be further improved by the analyst by incorporating
more historical data or engaging more domain experts whenever needed.

Our two implementations demonstrate that the methodology is practical
and it works on real attack trees. In general, the non-linear constraint
satisfaction problem used by the CSP-based implementation and the SQP problem 
are
NP-hard \cite{manyem2012computational,T1993}. Note that the complexity of
the attack-tree decoration task itself is unknown. Therefore, our
implementations will not scale to very large attack trees. To the best
of our knowledge, there is no statistics on the average or maximal
sizes of attack trees in practice. However, manually designed attack
trees, in our own experience, rarely have more than 100 nodes, because
they quickly become incomprehensible for humans~\cite{gadyatskaya2017new}. 
Therefore, we expect that our
implementations will work reasonably well with the majority of attack
trees created by practitioners.


\section{Conclusions}\label{sec:conclusions}
In this article we proposed the first quantitative attribute approach for 
attack trees that is able to deal with incomplete information. On the one hand, 
we have well-known classical attributes on attack trees expressing relations 
between nodes in a tree. On the other hand, we recognise that 
reliable information can be obtained from historical data and
domain knowledge. This type of information is typically not included
in the semantics of attack trees, and so it has
been largely ignored until now. 
We take these two views on the world and we
verify that they are consistent. Since Schneier's definition of attack trees, 
bottom-up evaluation
of attributes was the norm. We are the first to introduce a complete
view on attribute values, including missing values and approximation. 

The main benefit of our computational methodology is that it allows to obtain a 
consistent valuation for all attack tree nodes, even if some leaf nodes data is 
missing. This is not possible with the standard bottom-up decoration approaches.
We have shown that the distinction between
hard and soft constraints can be handy, as it allows the analyst to better 
understand the solution space and to interactively engage with the decoration 
problem. 
Lastly, 
we provide two implementations and we demonstrate the feasibility of the 
suggested approach on a case study. Our proof-of-concept implementations 
demonstrate the 
viability of the method, and they can be easily introduced into the established 
tools working with attack trees, such as the SecurITree tool and the ADTool.

\section*{Acknowledgments}
\noindent
The research leading
to these results has received funding from the European Union Seventh Framework
Programme under grant agreement number 318003 (TREsPASS), and from
the Fonds National de la Recherche Luxembourg under grant 
C13/IS/5809105 (ADT2P).

\bibliographystyle{plain}

\begin{thebibliography}{10}

\bibitem{ahmed2007review}
Ammar Ahmed, Berman Kayis, and Sataporn Amornsawadwatana.
\newblock A review of techniques for risk management in projects.
\newblock {\em Benchmarking: An International Journal}, 14(1):22--36, 2007.

\bibitem{amenaza}
Amenaza.
\newblock Securitree software, 2017.

\bibitem{Arnold-POST14}
Florian Arnold, Holger Hermanns, Reza Pulungan, and Mari{\"e}lle Stoelinga.
\newblock {Time-Dependent Analysis of Attacks}.
\newblock In {\em Proc. 3rd Int. Conf. on Principles of Security and Trust
  (POST'14)}, volume 8414 of {\em LNCS}, pages 285--305. Springer, 2014.

\bibitem{Aslanyan-POST-2015}
Z.~Aslanyan and F.~Nielson.
\newblock Pareto efficient solutions of attack-defence trees.
\newblock In {\em Proc. 4th Int. Conf. on Principles of Security and Trust
  (POST'15)}, volume 9036 of {\em LNCS}, pages 95--114. Springer, 2015.

\bibitem{aslanyan2016quantitative}
Zaruhi Aslanyan, Flemming Nielson, and David Parker.
\newblock Quantitative verification and synthesis of attack-defence scenarios.
\newblock In {\em Proc. 29th IEEE Computer Security Foundations Symposium
  (CSF'16)}, pages 105--119. IEEE, 2016.

\bibitem{aven2007unified}
Terje Aven.
\newblock A unified framework for risk and vulnerability analysis covering both
  safety and security.
\newblock {\em Reliability engineering \& System safety}, 92(6):745--754, 2007.

\bibitem{BaKoMeSc}
A.~Bagnato, B.~Kordy, P.~H. Meland, and P.~Schweitzer.
\newblock Attribute decoration of attack--defense trees.
\newblock {\em Int. J. of Sec. Soft. Engineering}, 3(2):1--35, 2012.

\bibitem{baker2007necessary}
Wade~H. Baker, Loren~P. Rees, and Peter~S. Tippett.
\newblock Necessary measures: metric-driven information security risk
  assessment and decision making.
\newblock {\em Communications of the ACM}, 50(10):101--106, 2007.

\bibitem{benini2008risk}
Marco Benini and Sabrina Sicari.
\newblock Risk assessment in practice: A real case study.
\newblock {\em Computer communications}, 31(15):3691--3699, 2008.

\bibitem{bistarelli2006defense}
Stefano Bistarelli, Fabio Fioravanti, and Pamela Peretti.
\newblock Defense trees for economic evaluation of security investments.
\newblock In {\em Proc. 1st Int. Conf. on Availability, Reliability and
  Security (ARES'06)}. IEEE, 2006.

\bibitem{bohme2010security}
Rainer B{\"o}hme.
\newblock Security metrics and security investment models.
\newblock In {\em Proc. 5th Int. Workshop on Security (IWSEC'10)}, volume 6434
  of {\em LNCS}, pages 10--24. Springer, 2010.

\bibitem{DBLP:conf/gamesec/BuldasL13}
Ahto Buldas and Aleksandr Lenin.
\newblock New efficient utility upper bounds for the fully adaptive model of
  attack trees.
\newblock In {\em Proc. 4th Int. Conf. on Decision and Game Theory for Security
  (GameSec'13)}, volume 8252 of {\em LNCS}, pages 192--205. Springer, 2013.

\bibitem{OCTAVE-2007}
R.~Caralli, J.~Stevens, L.~Young, and W.~Wilson.
\newblock Introducing {OCTAVE Allegro: Improving} the information security risk
  assessment process.
\newblock Technical Report CMU/SEI-2007-TR-012, Software Engineering Institute,
  Carnegie Mellon University, 2007.

\bibitem{dacier1996models}
Marc Dacier, Yves Deswarte, and Mohamed Ka{\^a}niche.
\newblock Models and tools for quantitative assessment of operational security.
\newblock In {\em Proc. IFIP Int. Conf. on ICT Systems Security and Privacy
  Protection (SEC'96)}, IFIPAICT, pages 177--186. Springer, 1996.

\bibitem{de2017using}
M.~H. de~Bijl.
\newblock Using data analysis to enhance attack trees.
\newblock In {\em Proc. Twente Student Conference}, 2017.

\bibitem{MouraB08}
Leonardo~Mendon{\c{c}}a de~Moura and Nikolaj Bj{\o}rner.
\newblock {Z3:} an efficient {SMT} solver.
\newblock In {\em {TACAS}}, volume 4963 of {\em Lecture Notes in Computer
  Science}, pages 337--340. Springer, 2008.

\bibitem{PoEM-2016}
Marlon Fraile, Margaret Ford, Olga Gadyatskaya, Rajesh Kumar, Mari{\"{e}}lle
  Stoelinga, and Rolando Trujillo{-}Rasua.
\newblock Using attack-defense trees to analyze threats and countermeasures in
  an {ATM:} {A} case study.
\newblock In {\em Proc. 9th IFIP Working Conference on the Practice of
  Enterprise Modeling (PoEM'16)}, Lecture Notes in Business Information
  Processing, pages 326--334. Springer, 2016.

\bibitem{gadyatskaya2017new}
O.~Gadyatskaya and R.~Trujillo-Rasua.
\newblock New directions in attack tree research: {Catching} up with industrial
  needs.
\newblock In {\em Proc. of GraMSec}, pages 115--126. Springer, 2017.

\bibitem{gadyatskaya2016modelling}
Olga Gadyatskaya, Ren{\'e}~Rydhof Hansen, Kim~Guldstrand Larsen, Axel Legay,
  Mads~Chr. Olesen, and Danny~B{\o}gsted Poulsen.
\newblock Modelling attack-defense trees using timed automata.
\newblock In {\em Proc. Int. Conf. on Formal Modeling and Analysis of Timed
  Systems (FORMATS'16)}, volume 9884 of {\em LNCS}, pages 35--50. Springer,
  2016.

\bibitem{gadyatskaya2016bridging}
Olga Gadyatskaya, Carlo Harpes, Sjouke Mauw, C{\'e}dric Muller, and Steve
  Muller.
\newblock Bridging two worlds: reconciling practical risk assessment
  methodologies with theory of attack trees.
\newblock In {\em Proc. 3rd Int. Workshop on Graphical Models for Security
  (GraMSec'16)}, volume 9987 of {\em LNCS}, pages 80--93. Springer, 2016.

\bibitem{gadyatskaya2016attack}
Olga Gadyatskaya, Ravi Jhawar, Piotr Kordy, Karim Lounis, Sjouke Mauw, and
  Rolando Trujillo-Rasua.
\newblock Attack trees for practical security assessment: ranking of attack
  scenarios with adtool 2.0.
\newblock In {\em Proc. 13th Int. Conf. on Quantitative Evaluation of Systems
  (QEST'16)}, volume 9826 of {\em LNCS}, pages 159--162. Springer, 2016.

\bibitem{hong2017survey}
Jin~B. Hong, Dong~Seong Kim, Chun-Jen Chung, and Dijiang Huang.
\newblock A survey on the usability and practical applications of graphical
  security models.
\newblock {\em Computer Science Review}, 26:1--16, 2017.

\bibitem{horne2017semantics}
Ross Horne, Sjouke Mauw, and Alwen Tiu.
\newblock Semantics for specialising attack trees based on linear logic.
\newblock {\em Fundamenta Informaticae}, 153(1-2):57--86, 2017.

\bibitem{jaquith2007security}
Andrew Jaquith.
\newblock {\em Security metrics}.
\newblock Pearson Education, 2007.

\bibitem{jhawar2015attack}
Ravi Jhawar, Barbara Kordy, Sjouke Mauw, Sa{\v{s}}a Radomirovi{\'c}, and
  Rolando Trujillo-Rasua.
\newblock Attack trees with sequential conjunction.
\newblock In {\em Proc. IFIP TC-11 Int. Information Security and Privacy
  Conference (IFIPSec'15)}, volume 455 of {\em IFIPAICT}, pages 339--353.
  Springer, 2015.

\bibitem{jhawar2016stochastic}
Ravi Jhawar, Karim Lounis, and Sjouke Mauw.
\newblock A stochastic framework for quantitative analysis of attack-defense
  trees.
\newblock In {\em Proc.\ 12th Workshop on Security and Trust Management
  (STM'16)}, volume 9871 of {\em LNCS}, pages 138--153. Springer, 2016.

\bibitem{KoMaSc-2012}
B.~Kordy, S.~Mauw, and P.~Schweitzer.
\newblock Quantitative questions on attack-defense trees.
\newblock In {\em Proc.\ 15th Annual International Conference on Information
  Security and Cryptology (ICISC'12)}, volume 7839 of {\em LNCS}, pages 49--64.
  Springer, 2013.

\bibitem{KordyKMS-QEST-13}
Barbara Kordy, Piotr Kordy, Sjouke Mauw, and Patrick Schweitzer.
\newblock {ADTool}: Security analysis with attack--defense trees.
\newblock In {\em Proc.\ 10th International Conference on Quantitative
  Evaluation of SysTems (QEST'13)}, volume 8054 of {\em LNCS}, pages 173--176.
  Springer, 2013.

\bibitem{KoMaRaSc_JLC}
Barbara Kordy, Sjouke Mauw, Sa\v{s}a Radomirovi\'c, and Patrick Schweitzer.
\newblock {Attack--Defense Trees}.
\newblock {\em Journal of Logic and Computation}, 24(1):55--87, 2014.

\bibitem{kordy2014dag}
Barbara Kordy, Ludovic Pi{\`e}tre-Cambac{\'e}d{\`e}s, and Patrick Schweitzer.
\newblock {DAG}-based attack and defense modeling: {Don't} miss the forest for
  the attack trees.
\newblock {\em Computer science review}, 13:1--38, 2014.

\bibitem{kordy2012computational}
Barbara Kordy, Marc Pouly, and Patrick Schweitzer.
\newblock Computational aspects of attack--defense trees.
\newblock In {\em Int. Joint Conferences on Security and Intelligent
  Information Systems (SIIS'11)}, volume 7053 of {\em LNCS}, pages 103--116.
  Springer, 2012.

\bibitem{KoPoSc_iFM14}
Barbara Kordy, Marc Pouly, and Patrick Schweitzer.
\newblock A probabilistic framework for security scenarios with dependent
  actions.
\newblock In {\em Proc. Integrated Formal Methods (IFM'14)}, volume 8739 of
  {\em LNCS}, pages 256--271, 2014.

\bibitem{kumar2015quantitative}
Rajesh Kumar, Enno Ruijters, and Mari{\"e}lle Stoelinga.
\newblock Quantitative attack tree analysis via priced timed automata.
\newblock In {\em Proc. Int. Conf. on Formal Modeling and Analysis of Timed
  Systems (FORMATS'15)}, volume 9268 of {\em LNCS}, pages 156--171. Springer,
  2015.

\bibitem{lenin2014attacker}
A.~Lenin, J.~Willemson, and D.~P. Sari.
\newblock Attacker profiling in quantitative security assessment based on
  attack trees.
\newblock In {\em 19th Nordic Conference on Secure IT Systems (NordSec'14)},
  volume 8788 of {\em LNCS}, pages 199--212. Springer, 2014.

\bibitem{TUT:LeninPhD2015}
Aleksandr Lenin.
\newblock {\em Reliable and Efficient Determination of the Likelihood of
  Rational Attacks}.
\newblock PhD thesis, Tallinn University of Technology, 2015.
\newblock TUT Press.

\bibitem{DBLP:conf/gamesec/LeninB14}
Aleksandr Lenin and Ahto Buldas.
\newblock Limiting adversarial budget in quantitative security assessment.
\newblock In {\em Proc. 5th Int. Conf. on Decision and Game Theory for Security
  (GameSec'14)}, volume 8840 of {\em LNCS}, pages 155--174. Springer, 2014.

\bibitem{DBLP:conf/gamesec/LeninWC15}
Aleksandr Lenin, Jan Willemson, and Anton Charnamord.
\newblock Genetic approximations for the failure-free security games.
\newblock In {\em Proc. 6th Int. Conf. on Decision and Game Theory for Security
  (GameSec'15)}, volume 9406 of {\em LNCS}, pages 311--321. Springer, 2015.

\bibitem{mahmood2013fuzzy}
Yasser~A. Mahmood, Alireza Ahmadi, Ajit~Kumar Verma, Ajit Srividya, and Uday
  Kumar.
\newblock Fuzzy fault tree analysis: A review of concept and application.
\newblock {\em Int. J. of System Assurance Engineering and Management},
  4(1):19--32, 2013.

\bibitem{manyem2012computational}
Prabhu Manyem and Julien Ugon.
\newblock Computational complexity, {NP} completeness and optimization duality:
  A survey.
\newblock {\em Electronic Colloquium on Computational Complexity (ECCC)},
  19(9), 2012.

\bibitem{MaOo}
Sjouke Mauw and Martijn Oostdijk.
\newblock Foundations of attack trees.
\newblock In {\em Proc.\ 8th Int. Conf. on Information Security and Cryptology
  (ICISC'05)}, volume 3935 of {\em LNCS}, pages 186--198. Springer, 2006.

\bibitem{oppliger2015quantitative}
Rolf Oppliger.
\newblock Quantitative risk analysis in information security management: a
  modern fairy tale.
\newblock {\em IEEE Security \& Privacy}, 13(6):18--21, 2015.

\bibitem{potteiger2016software}
Bradley Potteiger, Goncalo Martins, and Xenofon Koutsoukos.
\newblock Software and attack centric integrated threat modeling for
  quantitative risk assessment.
\newblock In {\em Proc. Symposium and Bootcamp on the Science of Security
  (HotSos'16)}, pages 99--108. ACM, 2016.

\bibitem{roy2012scalable}
A.~Roy, D.~S. Kim, and K.~Trivedi.
\newblock Scalable optimal countermeasure selection using implicit enumeration
  on attack countermeasure trees.
\newblock In {\em Proc. IEEE/IFIP Int. Conf. on Dependable Systems and Networks
  (DSN'12)}. IEEE, 2012.

\bibitem{ruijters2015fault}
Enno Ruijters and Mari{\"e}lle Stoelinga.
\newblock Fault tree analysis: a survey of the state-of-the-art in modeling,
  analysis and tools.
\newblock {\em Computer Science Review}, 15:29--62, 2015.

\bibitem{sarabi2015prioritizing}
Armin Sarabi, Parinaz Naghizadeh, Yang Liu, and Mingyan Liu.
\newblock Prioritizing security spending: A quantitative analysis of risk
  distributions for different business profiles.
\newblock In {\em Proc. 14th Annual Workshop on the Economics of Information
  Security (WEIS'15)}, 2015.

\bibitem{Schn}
Bruce Schneier.
\newblock {Attack Trees}.
\newblock {\em Dr. Dobb's Journal of Software Tools}, 24(12):21--29, 1999.

\bibitem{schneier2011secrets}
Bruce Schneier.
\newblock {\em Secrets \& Lies: Digital Security in a Networked World}.
\newblock John Wiley \& Sons, Inc., New York, NY, USA, 2000.

\bibitem{shostack2014threat}
A.~Shostack.
\newblock {\em Threat modeling: Designing for security}.
\newblock John Wiley \& Sons, 2014.

\bibitem{T1993}
Edward P.~K. Tsang.
\newblock {\em Foundations of constraint satisfaction.}
\newblock Computation in cognitive science. Academic Press, 1993.

\bibitem{vose2008risk}
David Vose.
\newblock {\em Risk analysis: a quantitative guide}.
\newblock John Wiley \& Sons, 2008.

\end{thebibliography}

\end{document}